\documentclass[aps,pra,reprint,nofootinbib,superscriptaddress]{revtex4-2}

\usepackage{color}
\usepackage{xcolor}

\usepackage{amsmath}
\usepackage{amssymb}
\usepackage{mathrsfs}
\usepackage{braket}
\usepackage{bm}
\usepackage{bbm}
\usepackage{multirow}

\usepackage{latexsym}

\usepackage{amsthm}
\theoremstyle{plain}
\newtheorem{thm}{Theorem}
\newtheorem{lem}{Lemma}

\newtheorem{prop}{Proposition}
\newtheorem*{thm**}{Theorem}
\newtheorem*{lem**}{Lemma}
\newtheorem*{cor**}{Corollary}
\newtheorem*{prop**}{Proposition}

\theoremstyle{definition}
\newtheorem{defi}{Definition}
\newtheorem*{defi**}{Definition}

\theoremstyle{remark}
\newtheorem{eg}{Example}
\newtheorem*{eg**}{Example}
\newtheorem{rmk}{Remark}
\newtheorem*{rmk**}{Remark}

\usepackage{graphicx}
\usepackage{bmpsize}
\usepackage[table,xcdraw]{xcolor}

\usepackage{hyperref} 
\hypersetup{
	colorlinks=true,
	linkcolor=blue,
	citecolor=blue,
	unicode=true,          
}
\usepackage[capitalize,nameinlink]{cleveref}

\newcommand{\1}{\mbox{1}\hspace{-0.25em}\mbox{l}}

\newcommand{\R}{\mathbb{R}}
\newcommand{\C}{\mathbb{C}}

\newcommand{\Tr}[1]{\mathrm{Tr}\!\left[{#1}\right]\!}

\newcommand{\ketbra}[2]{\ket{#1}\hspace{-0.25em}\bra{#2}}

\newcommand{\LL}{\mathcal{L}}

\newcommand{\hi}{\mathcal{H}} 
 
\newcommand{\id}{\mathbbm{1}}

\newcommand{\A}{\mathsf{A}}
\newcommand{\B}{\mathsf{B}}
\newcommand{\D}{\mathsf{D}}
\newcommand{\E}{\mathsf{E}}
\newcommand{\M}{\mathsf{M}}
\newcommand{\T}{\mathsf{T}}
\newcommand{\vsigma}{\boldsymbol{\sigma}} 
\newcommand{\va}{\mathbf{a}} 
\newcommand{\vb}{\mathbf{b}} 
\newcommand{\vc}{\mathbf{c}} 
\newcommand{\vg}{\mathbf{g}} 
\newcommand{\vm}{\mathbf{m}} 
\newcommand{\vn}{\mathbf{n}} 
\newcommand{\vs}{\mathbf{s}} 
\newcommand{\vx}{\mathbf{x}} 
\newcommand{\vy}{\mathbf{y}}
\newcommand{\I}{\mathcal{I}}
\newcommand{\state}{\mathcal{S}} 
\newcommand{\vz}{\mathbf{z}}

\newcommand{\pTr}[2]{\mathrm{Tr}_{#1}[#2]}

\usepackage{comment}

\begin{document}  
	\title{Comparing quantum incompatibility of device sets from an operational
    perspective}
	
	\author{Kensei Torii}
	\email[]{toriikensei@nagoya-u.jp}
    \affiliation{Department of Mathematical Informatics, Nagoya University, Furo-cho, Chikusa-ku, Nagoya 464-8601, Japan}
    
	\author{Ryo Takakura}
	\email[]{takakura.ryo.qiqb@osaka-u.ac.jp}
	\affiliation{Center for Quantum Information and Quantum Biology, The University of Osaka, 1-2 Machikaneyama, Toyonaka,
Osaka 560-0043, Japan\looseness=-1}
    \affiliation{Graduate School of Science, The University of Osaka, 1-1 Machikaneyama, Toyonaka,
Osaka 560-0043, Japan\looseness=-1}

	\author{Ryotaro Imamura}
	\affiliation{Department of Nuclear Engineering, Kyoto University, Kyoto daigaku-katsura, Nisikyo-ku, Kyoto 615-8540, Japan\looseness=-1}

	\begin{abstract}
        To effectively utilize quantum incompatibility as a resource in quantum information processing, it is crucial to evaluate how incompatible a set of devices is. 
        In this study, we propose an ordering to compare incompatibility and reveal its various properties based on the operational intuition that larger incompatibility can be detected with fewer states.
        We especially focus on typical class of incompatibility exhibited by mutually unbiased qubit observables and numerically demonstrate that the ordering yields new classifications among sets of devices.
        Moreover, the equivalence relation induced by this ordering is proved to uniquely characterize mutually unbiased qubit observables among all pairs of unbiased qubit observables.
        The operational ordering also has a direct implication for a specific protocol called distributed sampling. 
	\end{abstract}	
	
	\maketitle

	\section{Introduction}
    One of the most distinctive features of quantum theory is that certain physical operations are impossible to perform simultaneously.
    This property is termed \textit{incompatibility} and describes several fundamental concepts in quantum theory \cite{Heinosaari_2016,RevModPhys.95.011003}: uncertainty relations (precisely measurement uncertainty relations) \cite{Heisenberg1927,PhysRevA.53.2038,PhysRevA.67.042105,PhysRevA.78.052119,PhysRevLett.112.050401,10.1063/1.4871444} are quantitative expressions of incompatibility for observables and the no-cloning theorem (the no-broadcasting theorem) \cite{Wootters1982,DIEKS1982271,PhysRevLett.76.2818,PhysRevLett.79.2153,FAN2014241} is an example of incompatibility for channels.
    Besides its fundamental significance, incompatibility is recognized as a valuable resource in quantum information processing \cite{PhysRevLett.103.230402,PhysRevA.98.012126,PhysRevLett.122.130402, PhysRevLett.122.130403,PhysRevLett.122.130404, Oszmaniec2019operational, 10.1063/1.5126496, PhysRevA.101.032331, PhysRevA.101.052306, PhysRevLett.124.120401, m7ln-tb1s}, for example playing a crucial role in quantum cryptography and quantum random access code \cite{RevModPhys.74.145,PhysRevLett.98.230501,BENNETT20147,Carmeli_2020, PhysRevA.107.062210}. It suggests the necessity to evaluate the degree of incompatibility of a set of quantum devices (or simply, a \textit{device set}) such as observables or channels, and various studies have been given in this direction \cite{Busch_2013, HEINOSAARI20141695, PhysRevA.89.062112, Haapasalo_2015, PhysRevA.92.022115, PhysRevA.96.022113, Designolle_2019,Mordasewicz_2022, PhysRevA.98.012126}.

    From an operational perspective, states spanning the whole state space are necessary to test whether a device set is incompatible.
    A previous study \cite{PhysRevA.104.032228} introduced a notion relaxing this necessity because one often encounters practical situations where only restricted states are available.
    There a set of devices is called $\state_0$-\textit{incompatible} if we can detect its incompatibility using only states in a subset $\state_0$ of the whole state space.
    Based on this concept, two quantifications of incompatibility, termed \textit{incompatibility dimension} and \textit{compatibility dimension}, were proposed.  
    The incompatibility dimension of a device set is defined as the minimum number of (affinely independent) states needed to detect incompatibility, while the compatibility dimension is the maximum number of states which may be necessary to detect incompatibility. 
    However, these quantities only take into account the affine dimension of the subset $\state_0$ and do not reflect its structures.
    In other words, if we concentrate on a specific subset of states, the (in)compatibility dimension may not work as an appropriate measure of incompatibility.
    
    In this paper, keeping motivated by the same operational intuition as \cite{PhysRevA.104.032228}, we propose an ordering among device sets that enables more fine-grained comparisons of their incompatibility.
    The key idea is that a device set is considered more incompatible if its incompatibility is easier to detect. 
    More precisely, if a device set $\mathbb{D}$ is $\state_0$-incompatible for every $\state_0$ such that another device set $\mathbb{E}$ is $\state_0$-incompatible, then $\mathbb{D}$ is regarded as ``more incompatible" than $\mathbb{E}$. 
    This ordering preserves primitive properties of incompatibility: the incompatibility never increases under convex combinations with compatible devices, classical post-processing of observables, and concatenation of channels. 
    Beyond this consistency, the ordering yields novel classifications of device sets, offering structural insights into the concept of incompatibility.
    We further explore the equivalence relation induced by the ordering. 
    In particular, we focus on mutually unbiased qubit observables as one of the most typically incompatible devise sets, and reveal that they are uniquely distinguishable from other pairs of unbiased qubit observables through this equivalence.
    Also, the ordering describes a benefit in a specific informational task called distributed sampling \cite{PhysRevA.100.042308}.

     This paper is organized as follows. In Section \ref{sec_preliminaries}, we review the concepts of incompatibility and $\state_0$-incompatibility, accompanied by several examples. We then introduce our main proposal, an operational ordering of incompatibility, and present its general properties in Section \ref{subsec_ordering}. Section \ref{subsec_dist_sampling} illustrates the connection between this ordering and an informational task. In Section \ref{sec_qubit}, we provide detailed investigations on mutually unbiased qubit observables. We first explore the equivalence relation in Section \ref{subsec_equivalence}.
     We then numerically demonstrate that the ordering leads to new classifications in Section \ref{subsec_numerical_analysis}.  Finally in Section \ref{sec_conclusion}, we summarize this paper and outline potential directions for future work.

	\section{Preliminaries} \label{sec_preliminaries}
	Throughout this article, we only consider finite-dimensional Hilbert spaces.
	For such a Hilbert space $\hi$, we denote the set of all linear operators on $\hi$ by $\LL(\hi)$, and the set of all quantum states (the \textit{state space}) associated with $\hi$ by $\state(\hi)$, i.e.,
    \[
	\LL(\hi)=\{A\colon\hi\to\hi\mid\mbox{linear}\},
	\]
	and
	\[
	\state(\hi)=\{\rho\in\LL(\hi)\mid \rho\ge \bm{0}, \Tr{\rho}=1\}.
	\]
	An observable with a finite outcome set $X$ is described by a positive operator-valued measure (POVM) $\A = \{ \A(x) \}_{x\in X}$ such that
    \[
    \A(x) \geq \bm{0},\,\sum_{x\in X} \A(x) = \id.
    \]
    For a measurement of $\A$ on a system in a state $\rho$, the probability of obtaining an outcome $x$ is given by $\Tr{\rho\A(x)}\,$.
    A transformation from a state on $\hi^{in}$ into another state on $\hi^{out}$ is described by a completely positive and trace-preserving (CPTP) map or a \textit{channel} $\Lambda: \state(\hi^{in})\to\state(\hi^{out})$. 
    Note that in this paper we mainly use the Schr\"{o}dinger picture.
    A measurement process that provides an outcome probability distribution on $X$ and a post-measurement state is represented by an operation-valued measure or an \textit{instrument} $\mathcal{I}=\{\mathcal{I}_x\}_{x\in X}$ such that
    \[
    \bm{0} \leq \mathcal{I}_x(\rho)\leq \1,\, \sum_{x\in X}\Tr{\mathcal{I}_x(\rho)}=1,
    \]
    for all $\rho\in\state(\hi^{in})$. Each $\mathcal{I}_x$ is a completely positive and trace non-increasing map, and the sum $\sum_x \mathcal{I}_x$ is a channel. For an input state $\rho$, the quantity $\Tr{\mathcal{I}_x(\rho)}$ represents the probability of obtaining the outcome $x$, and the normalized state $\mathcal{I}_x(\rho)/\Tr{\mathcal{I}_x(\rho)}$ is the corresponding post-measurement state. 
    For more information on quantum measurement theory, we refer the reader to Refs. \cite{Heinosaari_Ziman_2011,Busch2016Quantum}.

    These three types of objects can be described comprehensively as maps acting on $\state(\hi^{in})$. 
    In fact, an observable $\A$ can be identified with a map $\rho \mapsto \{\Tr{\rho\A(x)}\}_x$ whose output is a probability distribution. A channel $\Lambda$ clearly lies within this framework. Also, an instrument $\mathcal{I}$ can be rewritten as a map $\rho \mapsto \{\Tr{\mathcal{I}_x(\rho)}\otimes \mathcal{I}_x(\rho)\}_x$ whose output is a distribution of a probability and a post-measurement state.
    We describe these objects by the notion of \textit{devices} in a unified way. The input space of a device $\D$ is always $\state(\hi^{in})$, while the output space depends on the type of device; $\mathcal{P}(X)$ for an observable, $\state(\hi^{out})$ for a channel, and $\mathcal{P}(X)\otimes \state(\hi^{out})$ for an instrument. Here $\mathcal{P}(X)=\{P=\{P(x)\}_{x\in X}\mid 0\le P(x)\le1,\sum_x P(x)=1\}$ denotes the set of all probability distributions over $X$.

    In this paper, we will study \textit{device sets}, specifically sets of observables (\textit{observable sets}), sets of channels (\textit{channel sets}), and pairs of an observable and a channel (\textit{observable-channel pairs}).
    Now we present the definition of incompatibility.
	\begin{defi}\label{def_incomp}
		A device set $\mathbb{D}=\{\D_1, \ldots, \D_n\}$ is called \textit{compatible} if there exists a device $\D'$, called a \textit{joint device}, such that each device $\D_i$ is a marginal of $\D'$. The input space of $\D'$ is $\state(\hi^{in})$ and its output space is the tensor product of the output spaces of $\D_1,\ldots,\D_n$. 
        If no such joint device exists, the set $\mathbb{D}$ is called \textit{incompatible}. 
	\end{defi}

    We will present three examples illustrating this definition: the incompatibility of observable sets, channel sets, and observable-channel pairs.
    
\begin{eg}[Observable set] \label{eg:incomp_obs}
    First we illustrate the incompatibility of an observable set $\mathbb{A}=\{\A_1,\ldots,\A_n\}$ each with an outcome set $X_i$. The set $\mathbb{A}$ is compatible (or \textit{jointly measurable}) if and only if we can construct a joint observable $\A'$ with outcome set $X_1\times\cdots\times X_n$ such that
    \begin{align*}
        \Tr{\rho \A_1(x_1)}&=\sum_{x_2,\ldots,x_n}\Tr{\rho \A'(x_1,\ldots,x_n)},\\
        \Tr{\rho \A_2(x_2)} &=\sum_{x_1,x_3,\ldots,x_n}\Tr{\rho\A'(x_1,\ldots,x_n)},\\
        &\vdots\\
        \Tr{\rho \A_n(x_n)} &= \sum_{x_1,\ldots,x_{n-1}}\Tr{\rho \A'(x_1,\ldots,x_n)}
    \end{align*}
    for all $x_i\in X_i,~i=1,\ldots,n$ and $\rho \in \state(\hi)$. 
    Let us further study the case $\hi=\C^2$ (a single qubit system). A pair $\{\A_1,\A_2\}$ of qubit observables each with a binary output $\pm$ is described by POVMs of the form
    \begin{align}
        \A_i(+) = \frac{1}{2}\big(a_i^0 \id + \va_i \cdot \vsigma\big),\, \A_i(-) = \id - \A_i(+) \label{def_QO}
    \end{align}
    for $i=1,\,2$. Here $\va_i \in \R^3$ is a vector satisfying $|\va_i|\leq a_i^0\leq 2-|\va_i|$ and $\vsigma=(\sigma_x,\sigma_y,\sigma_z)^{\T}$ is a vector of Pauli operators.
    The necessary and sufficient condition for $\{\A_1,\A_2\}$ to be compatible is \cite{PhysRevA.81.062116} 
    \begin{align}
        F &=: \min \biggl\{ 1\!-\!|\alpha|, 1\!-\!|\beta|, \notag \\
        & \sum_{\nu=\pm} |\va_1 \!+\! \va_2 \!+\! \nu \vg| \!+\! \sum_{\nu=\pm} |\va_1 \!-\! \va_2 \!+\! \nu \vg|-4 \biggr\} \leq  0, \label{cond_comp}
    \end{align}
    where
    \begin{align*}
        \gamma &= \va_1 \cdot \va_2 - (a_1^0 -1)(a_2^0 -1),\\
        \alpha &= \frac{(a_2^0 + \gamma a_1^0- \gamma -1)|\va_2|^2 - (a_1^0 + \gamma a_2^0 - \gamma -1) \va_1 \cdot \va_2}{|\va_1 \times \va_2|^2} ,\\
        \beta &= \frac{(a_1^0 + \gamma a_2^0- \gamma -1)|\va_1|^2 - (a_2^0 + \gamma a_1^0 - \gamma -1) \va_1 \cdot \va_2}{|\va_1 \times \va_2|^2} ,\\
        \vg &= \alpha \va_1 + \beta \va_2.
    \end{align*}
    Equivalent conditions are also derived in Refs.\cite{PhysRevA.78.012315,Busch2010}.
    A special type of \eqref{def_QO} is an \textit{unbiased qubit observable} \cite{Busch2009-2}, denoted by $\A_\mathrm{ub}^{\va}$.
    It is given by setting $a^0=1$ in \eqref{def_QO}, i.e.,
    \begin{align}
    \A^{\va}_\mathrm{ub}(\pm) = \frac{1}{2}\big(\id \pm \va \cdot \vsigma\big). \label{def_UQO}
    \end{align}
    Here $\va \in \R^3$ is a vector with $|\va|\leq 1$. 
    The magnitude $|\va|$ represents the sharpness of the measurement with smaller values indicating more noise. 
    A pair $\mathbb{A}_\mathrm{ub}:=\{\A^{\va_1}, \A^{\va_2}\}$ of unbiased qubit observables is compatible if and only if \cite{PhysRevD.33.2253}
    \begin{align}
        |\va_1 + \va_2| + |\va_1 - \va_2| \leq 2.\label{cond_comp_UQO}
    \end{align}
    If two vectors $\va_1$ and $\va_2$ are orthogonal, the pair $\mathbb{A}_\mathrm{ub}$ is specifically called \textit{mutually unbiased} \cite{Wootters1986}.
    For mutually unbiased qubit observables $\mathbb{A}_\mathrm{mub}^t :=\{\A_\mathrm{ub}^{t\vx},\A_\mathrm{ub}^{t\vy}\}$, where $\vx=(1,0,0)^\mathsf{T},\,\vy=(0,1,0)^\mathsf{T}$ in $\R^3$ and $t\in [0,1]$, condition \eqref{cond_comp_UQO} reduces to $t\leq \frac{1}{\sqrt{2}}$.
    \end{eg}

    \begin{eg}[Channel set]
    Let $\bm{\Lambda}=\{\Lambda_1,\ldots,\Lambda_n\}$ be a channel set with a common input space $\state(\hi^{in})$ but potentially different output spaces $\state(\hi^{out}_1),\ldots,\state(\hi^{out}_n)$. The set $\bm{\Lambda}$ is compatible if and only if we can construct a joint channel $\Lambda':\state(\hi^{in})\to\state(\hi^{out}_1)\otimes\cdots\otimes\state(\hi^{out}_n)$ satisfying
    \begin{align*}
        \Lambda_1(\rho) &= \pTr{2,\ldots,n}{\Lambda'(\rho)},\\
        \Lambda_2(\rho)&=\pTr{1,3,\ldots,n}{\Lambda'(\rho)},\\
        &\vdots\\
        \Lambda_n(\rho) &= \pTr{1,\ldots,n-1}{\Lambda'(\rho)}
    \end{align*}
    for all $\rho\in\state(\hi^{in})$ \cite{Heinosaari_2017}. 
    Here $\mathrm{Tr}_i$ denotes the partial trace over the $i$th Hilbert space. 
    As a typical example, we now consider the channel pair $\{\mathsf{id},\mathsf{id}\}$, where $\mathsf{id}\colon\rho\mapsto\rho$ is the identity channel on $\state(\hi)$. 
    The pair $\{\mathsf{id},\mathsf{id}\}$ is known to be incompatible (the \textit{no-broadcasting theorem} \cite{PhysRevLett.76.2818}), which means that there does not exist a joint channel $\mathsf{id}'\colon \state(\hi)\to\state(\hi)\otimes\state(\hi)$ such that
     \[
     \pTr{i}{\mathsf{id}'(\rho)}=\mathsf{id}(\rho)=\rho \quad(i=1,2)
     \]
    for all $\rho\in\state(\hi)$.
    \end{eg}

    \begin{eg}[Observable-channel pair] 
    Consider an observable $\A=\{\A(x)\}_{x\in X}$ and a channel $\Lambda$ with the same input space $\state(\hi^{in})$.
    When the pair $\{\A,\Lambda\}$ is compatible, their joint device is an instrument $\I=\{\I_x\}_{x\in X}$ such that 
    \begin{align}
    \Tr{\rho\A(x)}= \Tr{\I_x(\rho)}\,,~\Lambda(\rho)=\sum_{x\in X}\I_x(\rho), \label{joint_device_ins}
    \end{align}
    for all $x\in X$ and $\rho\in\state(\hi^{in})$. 
    We focus on an unbiased qubit observable $\A^{\va}_\mathrm{ub}$ and a depolarizing channel $\Lambda^{(t)}$, which is a simple model for measurement disturbance. 
    This channel $\Lambda^{(t)}$ is defined by
    \[
    \Lambda^{(t)}(\rho)=t\rho+(1-t)\frac{1}{2}\1
    \]
    for $t\in[0,1]$, where smaller $t$ represents larger disturbance.
    The pair $\{\A^{\va}_\mathrm{ub},\Lambda^{(t)}\}$ is compatible if and only if \cite{PhysRevA.97.022112}
    \[
    |\va|\le\frac{1}{2}\Big\{1-t +\sqrt{(1-t)(1+3t)}\Big\}.
    \]
    In the extreme case of $t=1$, the observable $\A_\mathrm{ub}^{\va}$ is compatible with $\Lambda^{(t=1)}=\mathsf{id}$ only if $\va$ is the zero vector and thus $\A_\mathrm{ub}^\va$ is a trivial observable. This means that we cannot get any information without disturbance \cite{Busch2009}.
\end{eg}

\begin{rmk}
    For a set of instruments, there are several definitions of incompatibility, such as traditional incompatibility \cite{Heinosaari2014}, parallel incompatibility \cite{PhysRevA.105.052202,Leppajarvi2024incompatibilityof}, and q-incompatibility \cite{Buscemi2023unifyingdifferent}. 
    Nonetheless, we will not discuss them in this paper.
\end{rmk}

    According to Definition \ref{def_incomp}, a set of states (a \textit{state set}) that spans the whole $\state(\hi^{in})$ is required to determine whether a device set is incompatible. 
    However, as mentioned in the introduction, we often faced with the situations where only restricted states are accessible.
    To address these situations, a relaxed notion of incompatibility was introduced in \cite{PhysRevA.104.032228}.
    
\begin{defi} \label{def:S0-incomp}
	Let $\state_0$ be a subset of $\state(\hi^{in})$.
	A device set $\mathbb{D}=\{\D_1, \ldots, \D_n\}$ is called \textit{$\state_0$-compatible} if there exists a compatible device set $\tilde{\mathbb{D}}=\{ \tilde{\D}_1, \ldots , \tilde{\D}_n \}$ such that
    \begin{align}
    \tilde{\D}_i(\rho) = \D_i(\rho) \label{def:S0-comp}
    \end{align}
    for all $i=1,\ldots, n$ and $\rho \in \state_0$. Otherwise, the set $\mathbb{D}$ is called \textit{$\state_0$-incompatible}.
\end{defi}	

A compatible device set $\mathbb{D}$ is clearly $\state_0$-compatible for any $\state_0\subset\state(\hi^{in})$. Conversely if $\mathbb{D}$ is $\state_0$-incompatible for some $\state_0$, we conclude that $\mathbb{D}$ is an incompatible device set.
In this sense, the $\state_0$-incompatibility of a device set is interpreted as the detectability of its incompatibility using only states in $\state_0$.
Thus we are interested in the cases where the set $\mathbb{D}$ is incompatible. 
Besides, in extreme cases the concept of $\state_0$-incompatibility becomes trivial. 
A state set $\state_0$ with only one state can never detect any incompatibility. 
Also, a state set $\state_0$ that spans the whole state space 
makes the notion of $\state_0$-incompatibility identical to standard incompatibility. 

Thanks to the linearity of \eqref{def:S0-comp}, the $\state_0$-incompatibility is equivalent to the $\bar{\state}_0$-incompatibility for the state set $\bar{\state}_0$ defined as
\[
\bar{\state}_0=\bigg\{\rho\in\state(\hi^{in}) \bigg|\, \rho=\sum_j c_j \rho_j,\, c_j\in\C,\,\rho_j\in\state_0\bigg\}.
\]
Since the state set $\bar{\state}_0$ is easier to handle than $\state_0$, we often regard state sets as affine subspaces of $\state(\hi^{in})$.

We exemplify the $\mathcal{S}_0$-incompatibility of a pair of qubit observables as a counterpart to Example \ref{eg:incomp_obs}. Henceforth, the state space $\state(\C^2)$ (the Bloch ball) is written by $\state$.

\begin{eg}[$\state_0$-incompatibility of qubit observables]\label{eg:S0-incomp_qubit_obs}
    Let us consider a subset $\state_0=\{\rho_j\}_j$ of $\state$, where each $\rho_j$ is characterized by a Bloch vector $\vs_j\in\R^3$ as $\rho_j=(\1+\vs_j\cdot \vsigma)/2$.
    Then the associated affine subspace $\bar{\state_0}$ of $\state$ is 
    \[
    \bar{\state}_0 \!=\! \bigg\{ \rho \in \state \bigg|\, \rho \!=\! \frac{1}{2}\Big(\1 + \big(\sum_j c_j \vs_j\big)\cdot \vsigma\Big) ,c_j \!\in\! \R, \sum_j c_j \!=\! 1 \bigg\}.
    \]
    Except for the trivial cases mentioned above, we consider $\state_0$ with two or three linearly independent states. If $\state_0$ has three linearly independent states (such $\state_0$ is denoted by $\state_0^{(3)}$), then $\bar{\state_0}$ is an intersection of a 2D plane and the Bloch ball. 
    This $\state_0^{(3)}$ is characterized by the distance $r \in [0,1]$ from the origin to the 2D plane and a normal vector $\vn \in \R^3$ (see FIG. \ref{fig1}). 
    If $\state_0$ has two linearly independent states (such $\state_0$ is denoted by $\state_0^{(2)}$), then $\bar{\state_0}$ is an intersection of a 1D line and the Bloch ball. This $\state_0^{(2)}$ is characterized by the distance $r \in [0,1]$ from the origin to the 1D line, a normal vector $\vn \in \R^3$ in the 2D plane spanned by the origin and $\state_0$, and the other normal vector $\vm \in \R^3$ orthogonal to the 2D plane (FIG. \ref{fig2}).
    \begin{figure}[t]
        \centering
        \includegraphics[width=0.8\columnwidth]{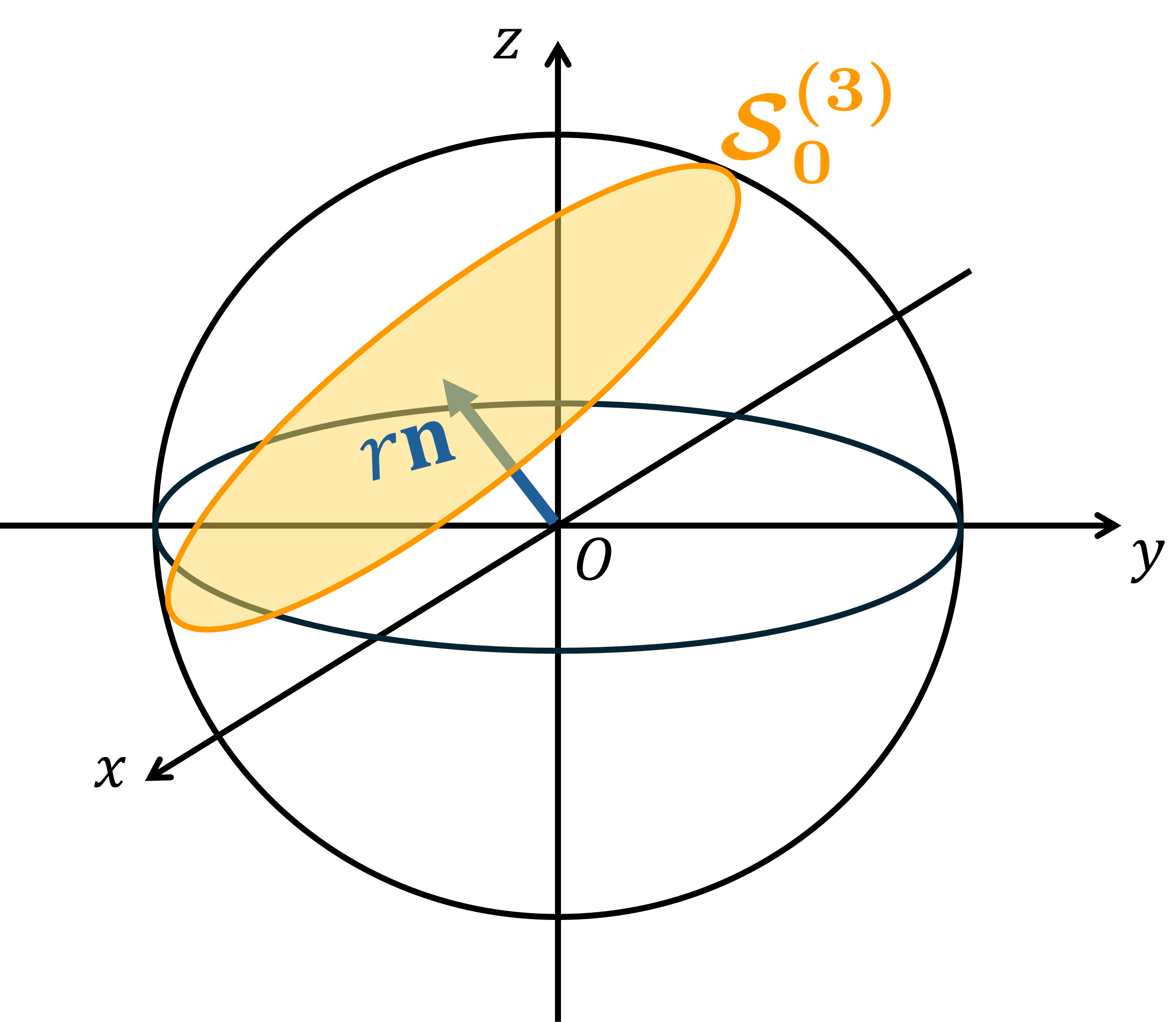}
        \caption{Parametrization of a set $\state_0$ consisting of three linearly independent states. Such $\state_0$ is identified with an intersection of a 2D plane and the Bloch ball, characterized by $r$ and $\vn$. The real parameter $r\in [0,1)$ denotes the distance from the origin to the intersection, and the vector $\vn \in \R^3$ is a normal vector.}
        \label{fig1}

        \centering
        \includegraphics[width=0.8\columnwidth]{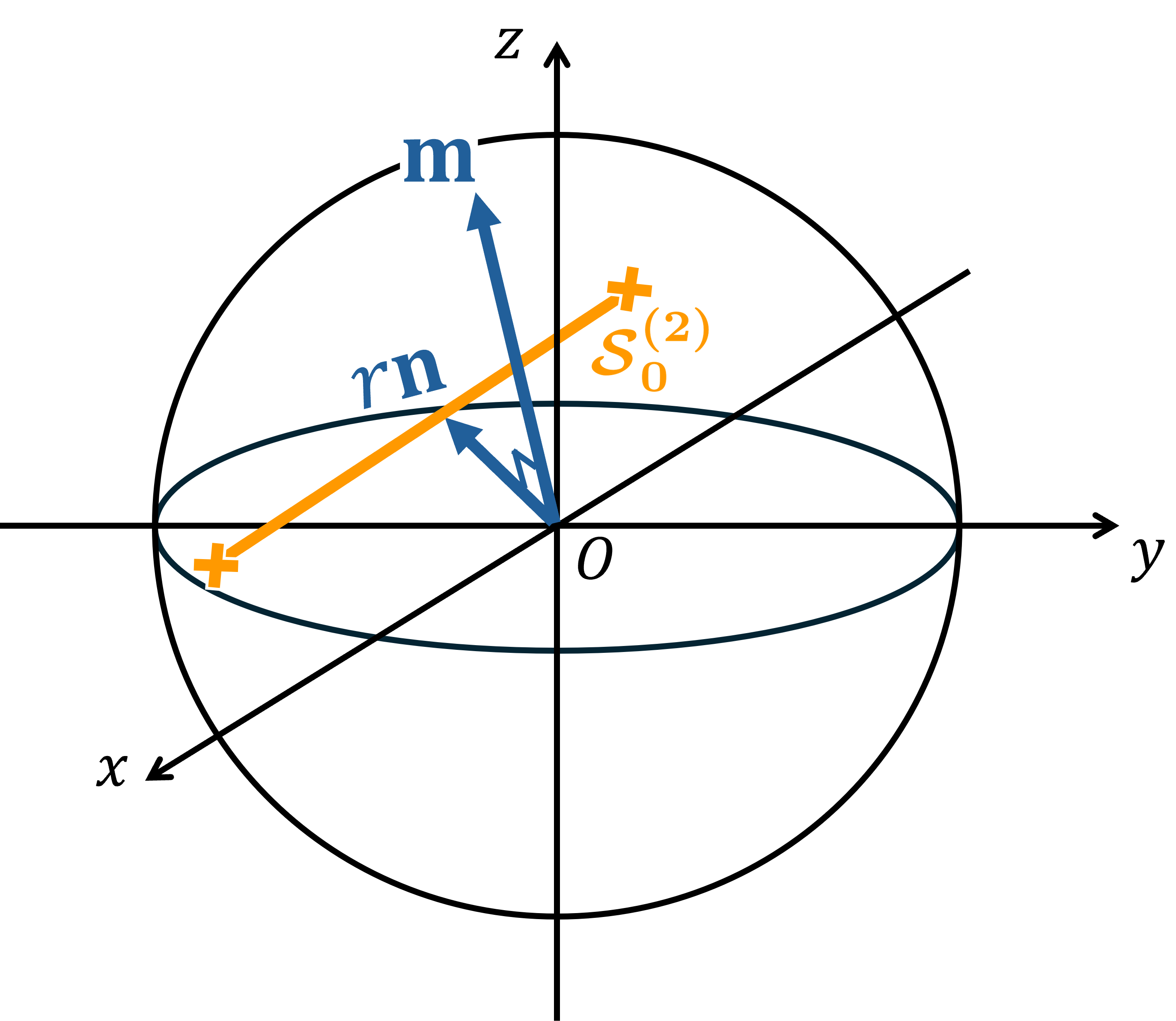}
        \caption{Parametrization of a set $\state_0$ consisting of two states. Such $\state_0$ is identified with an intersection of a 1D line and the Bloch ball, characterized by $r$, $\vn$ and $\vm$. The real parameter $r\in [0,1)$ denotes the distance from the origin to the intersection, the vector $\vn \in \R^3$ is a normal vector in the 2D plane spanned by the origin and $\state_0$, and the other normal vector $\vm\in\R^3$ is orthogonal to that 2D plane.}
        \label{fig2}
    \end{figure}
    By Definition \ref{def:S0-incomp}, the $\state_0$-incompatibility of $\{\A_1,\A_2\}$ is equivalent to the incompatibility of $\{\tilde{\A}_1,\tilde{\A}_2\}$ such that
    \begin{align}
     \mathrm{Tr}\big[\rho\tilde{\A}_i(\pm)\big]=\mathrm{Tr}\big[\rho\A_i(\pm)\big] \quad(i=1,2)\label{tilde_A}
    \end{align}
    for all $\rho\in\state_0$.
    If $\state_0$ has two linearly independent states, characterized by $r,\,\vn$ and $\vm$, then the observable $\tilde{\A}_i \equiv \tilde{\A}_i^{(r,\vn,\vm)}$ satisfying \eqref{tilde_A} has the form
    \begin{equation}
    \begin{aligned} 
        \tilde{\A}_i^{(r,\vn,\vm)}(+)&=\frac{1}{2}((a_i^0 -\lambda_i r)\1 +(\va_i+\lambda_i\vn+\xi_i\vm)\cdot\vsigma),\notag \\
        \tilde{\A}_i^{(r,\vn,\vm)}(-)&=\1-\tilde{\A}_i(+), 
    \end{aligned}
    \end{equation}
    where $\lambda_i,\, \xi_i$ are real parameters such that 
    \[|\va_i+\lambda_i\vn+\xi_i\vm|\le |a_i^0-\lambda_i r|\le2-|\va_i+\lambda_i\vn+\xi_i\vm|.\]
    Hence the necessary and sufficient condition for $\{\A_1,\A_2\}$ to be $\state_0$-compatible is 
    \begin{align}
        \min _{\lambda_1,\lambda_2,\xi_1,\xi_2}F_{\{\A_1,\A_2\}}^{(r,\vn,\vm)}(\lambda_1,\lambda_2,\xi_1,\xi_2) \geq 0. \label{cond_S0-comp_(2)}
    \end{align}
    Here $F_{\{\A_1,\A_2\}}^{(r,\vn,\vm)}$ is the left-hand side of \eqref{cond_comp} for two observables $\tilde{\A}_1^{(r,\vn,\vm)}$ and $\tilde{\A}_2^{(r,\vn,\vm)}$ and regarded as a function of $\lambda_1,\,\lambda_2,\,\xi_1$ and $\xi_2$.
    Similarly, if $\state_0$ has three linearly independent states, characterized by $r$ and $\vn$, then $\tilde{\A}_i\equiv\tilde{\A}_i^{(r,\vn)}$ is described as
    \begin{align*}
     \tilde{\A}_i^{(r,\vn)}(+)&=\frac{1}{2}((a_i^0 -\lambda_i r)\1 +(\va_i+\lambda_i\vn)\cdot\vsigma),\notag \\
     \tilde{\A}_i^{(r,\vn)}(-)&=\1-\tilde{\A}_i(+), 
    \end{align*}
    where $\lambda_i$ is a real parameter satisfying 
    \[|\va_i+\lambda_i\vn|\le |a_i^0-\lambda_i r|\le2-|\va_i+\lambda_i\vn|.\]
    Thus the counterpart of the condition \eqref{cond_S0-comp_(2)} is 
    \begin{align}
        \min _{\lambda_1,\lambda_2}F_{\{\A_1,\A_2\}}^{(r,\vn)}(\lambda_1,\lambda_2) \geq 0, \label{cond_S0-comp_(3)}
    \end{align}
     where $F_{\{\A_1,\A_2\}}^{(r,\vn)}$ is again the left-hand side of \eqref{cond_comp} for two observables $\tilde{\A}_1^{(r,\vn)}$ and $\tilde{\A}_2^{(r,\vn)}$ and regarded as a function of $\lambda_1$ and $\lambda_2$.
    There is a simple method for determining the $\state_0$-(in)compatibility for specific $\state_0$ and an observable pair. We especially denote a state set $\state_0^{(3)}$ including the origin by $\state_0^{(R)}$, i.e.,  
    \begin{equation}
    \begin{aligned} \label{def_S_R}
    &\mathcal{S}_0^{(R)} = \{ \rho = (\1 + \vs \cdot \vsigma)/2 |\, \vs \in R \}, \\
    &R = \{ \vs \in \R^3|\, |\vs| \leq 1, \vs \cdot \vn = 0 \}, ~\vn \in \R^3,~ |\vn|=1.
    \end{aligned}
    \end{equation}
    For a given $\state_0^{(R)}$, a pair $\{\A^{\va_1},\A^{\va_2}\}$ of unbiased qubit observables is $\state_0^{(R)}$-compatible if and only if
    \begin{align}
        |P_R(\va_1+\va_2)|+|P_R(\va_1-\va_2)|\leq2, \label{cond_S_R-comp}
    \end{align}
    where $P_R$ is the projection onto $R$ in $\R^3$. This is because the $\state_0$-incompatibility of $\{\A_\mathrm{ub}^{\va_1},\A_\mathrm{ub}^{\va_2}\}$ can be regarded as the incompatibility of their restrictions on $\state_0^{(R)}$, by treating the convex set $\state_0^{(R)}$ as a state space.     
\end{eg}

\begin{rmk}
   Regarding more general cases including higher-dimensional systems ($\dim \hi >2$), a larger number of observables, and other types of devices, it becomes more difficult to verify the $\state_0$-incompatibility. 
   First, for observables in higher-dimensional systems, the condition to judge their incompatibility as in \eqref{cond_comp} remains unknown.
   There is also a difficulty that the counterpart of $\tilde{\A}_i$ in \eqref{tilde_A} becomes more complicated because the Bloch representation in this case needs more parameters and puzzling restrictions \cite{KIMURA2003339}.
   Next, for a larger number of observables, the necessary and sufficient condition for three unbiased qubit observables to be compatible was revealed in Refs. \cite{Pal_2011, yu2013quantumcontextualityjointmeasurement}.
   However, we cannot apply this condition since the counterpart of $\tilde{\A}_i$ can be biased, and the analysis becomes similarly difficult for more than three observables.
   Finally, for channel pairs or observable-channel pairs, the general conditions for their compatibility are unknown even in a qubit system.
\end{rmk}

    From a viewpoint of detecting incompatibility with only restricted states, two quantifications of incompatibility were proposed \cite{PhysRevA.104.032228}.

\begin{defi}\label{incomp_dim}
    For an incompatible device set $\mathbb{D}=\{\D_1,\ldots,\D_n\}$, the incompatibility dimension $\chi_\mathrm{inc}$ and the compatibility dimension $\chi_\mathrm{com}$ are defined by
    \begin{align*}
        \chi_\mathrm{inc}(\mathbb{D})=\min_{\state_0\subset\state(\hi)} \big\{\dim \textrm{aff} \state_0 +1 |~\mathbb{D}\colon\state_0\text{-incompatible}\big\}
    \end{align*}
    and
    \begin{align*}
        \chi_\mathrm{com}(\mathbb{D})=\max_{\state_0\subset\state(\hi)} \big\{\dim \textrm{aff} \state_0 +1 |~\mathbb{D}\colon\state_0\text{-compatible}\big\}.
    \end{align*}
    Here $\dim \mathrm{aff} \state_0$ is the affine dimension of the affine hull of $\state_0$. 
\end{defi}
The incompatibility dimension is operationally interpreted as a minimum number of states needed to detect the incompatibility if one optimally chooses the available states.
On the other hand, the compatibility dimension represents the maximum number of (affinely independent) states one may need to detect the incompatibility.
These quantities $\chi_\mathrm{inc}$ and $\chi_\mathrm{com}$ take integer values, where smaller values indicate that the device set is more incompatible.
	
\section{Operational comparison of  incompatibility}\label{sec_ordering}
\subsection{Operational ordering of incompatibility}\label{subsec_ordering}
    We introduce a binary relation motivated by the concept of $\state_0$-incompatibility to compare incompatibility of device sets. This is the main object of this paper. 
	\begin{defi}
		For pairs of device sets $\mathbb{D}=\{\D_1, \ldots, \D_n\}$ and $\mathbb{E}=\{\E_1, \ldots, \E_m\}$, we write $\mathbb{E}\preceq_\mathrm{inc}\mathbb{D}$ if $\mathbb{D}$ is $\state_0$-incompatible for any state set $\state_0\subset\state(\hi^{in})$ such that $\mathbb{E}$ is $\state_0$-incompatible, or equivalently, $\mathbb{E}$ is $\state_0$-compatible for any $\state_0\subset\state(\hi^{in})$ such that $\mathbb{D}$ is $\state_0$-compatible.
        Clearly, this relation $\preceq_\mathrm{inc}$ is a pre-ordering of device sets.
	\end{defi}
	By definition, if a device set $\mathbb{E}$ is compatible, then $\mathbb{E}\preceq_\mathrm{inc}\mathbb{D}$ holds for any device set $\mathbb{D}$. Also, the relation implies that the (in)compatibility dimension of $\mathbb{D}$ is less than or equal to $\mathbb{E}$, i.e., $\chi_\mathrm{inc}(\mathbb{D})\le\chi_\mathrm{inc}(\mathbb{E})$ and $\chi_\mathrm{com}(\mathbb{D})\le\chi_\mathrm{com}(\mathbb{E})$. 
    Although the set $\mathbb{E}$ may have a different number of devices from $\mathbb{D}$, hereafter we usually consider that $\mathbb{D}$ and $\mathbb{E}$ have the same number of devices. 
    
    The ordering $\preceq_\mathrm{inc}$ is consistent with known properties of incompatibility.
    We start with a basic property that convex combination with a compatible device set does not increase the incompatibility.
    \begin{prop}\label{fund_features} 
        Let $\mathbb{D}=\{\D_1,\ldots,\D_n\},~\mathbb{E}=\{\E_1,\ldots,\E_n\}$ be device sets and let $\mathbb{N}=\{\mathsf{N}_1,\ldots,\mathsf{N}_n\}$ be a compatible device set.
        Suppose that $\mathbb{E}$ is a convex combination of $\mathbb{D}$ and $\mathbb{N}$, that is, there exists a real number $\lambda\in[0,1]$ such that
        \[
        \E_i = \lambda \D_i + (1-\lambda)\mathsf{N}_i
        \]
        for every $i=1,\ldots,n$. Then $\mathbb{E} \preceq_\mathrm{inc} \mathbb{D}$ holds.
    \end{prop}

    \begin{proof}
     For a subset $\state_0$ such that $\mathbb{D}$ is $\state_0$-compatible, we can find a compatible device set $\tilde{\mathbb{D}}=\{\tilde{\D}_1,\ldots,\tilde{\D}_n\}$ satisfying $\D_i(\rho)=\tilde{\D}_i(\rho)~(i=1,\ldots,n)$ for all $\rho\in\state_0$. 
     Denote by $\mathsf{N}'$ and $\tilde{\D}'$ the joint device of $\mathbb{N}$ and $\tilde{\mathbb{D}}$, respectively. Then the device $\E':=\lambda\tilde{\D}'+(1-\lambda)\mathsf{N}'$ gives each output $\E_i(\rho)$ as a marginal of $\E'(\rho)$ for $\rho\in\state_0$, which means that $\mathbb{E}$ is $\state_0$-compatible. 
    \end{proof}

    Furthermore, the ordering also preserves established properties of incompatibility for each types of device sets.
    First we focus on observable sets. The ordering properly expresses that post-processing of observables does not increase incompatibility \cite{RevModPhys.95.011003, PhysRevLett.122.130403}.     
	
	\begin{prop} \label{prop:post-processing}
        Let $\mathbb{A}=\{\A_1, \ldots, \A_n\}$ and $\mathbb{B} = \{\B_1, \ldots, \B_n\}$ be observable sets, where each $\A_i$ and $\B_i$ have outcome sets $X_i$ and $Y_i$, respectively.
        Suppose that each $\B_i$ is a post-processing of $\A_i$ via a stochastic matrix $p_i = \{p_i(y_i|x_i)\}_{(x_i, y_i)\in X_i \times Y_i}$, i.e., 
        \[
        \B_i(y_i) = \sum_{x_i\in X_i} p_i(y_i|x_i)\A_i(x_i)
        \]
        for every $y_i \in Y_i$ and $i=1,\ldots,n$. Then  $\mathbb{B} \preceq_\mathrm{inc} \mathbb{A}$ holds.
	\end{prop}
	\begin{proof} 
        For each $\state_0$ such that $\mathbb{A}$ is $\state_0$-compatible, there exists a compatible observable set $\tilde{\mathbb{A}}=\{\tilde{\A}_1,\ldots,\tilde{\A}_n\}$ satisfying
        \[
        \mathrm{Tr}\big[\rho\tilde{\A}_i(x_i)\big] = \mathrm{Tr}\big[\rho\A_i(x_i)\big]
        \]
        for all $x_i\in X_i$, $i=1,\ldots,n$ and $\rho\in\state_0$. 
        Let $\tilde{\A}'$ be a joint observable of $\tilde{\mathbb{A}}$. 
        If we define an observable $\tilde{\B}'$ with an outcome set $Y_1\times\cdots\times Y_n$ as
        \[
        \tilde{\B}'(y_1,\ldots,y_n)=\sum_{x_1,\ldots,x_n} \bigg(\prod_i p_i(y_i|x_i)\bigg)\tilde{\A}'(x_1,\ldots,x_n),
        \]
        then we have
        \[
        \Tr{\rho \B_i(y_i)}=\sum_{y_1,\ldots,y_{i-1},y_{i+1},\ldots,y_n} \Tr{\rho\tilde{\B}'(y_1,\ldots,y_n)}
        \]
        for all $y_i\in Y_i$, $i=1,\ldots,n$ and $\rho\in\state_0$. This gives the $\state_0$-compatibility of $\mathbb{B}$. 
	\end{proof}
    Next we consider channel sets.
    It is known that concatenations of channels do not increase incompatibility \cite{Heinosaari_2017}.
    A concatenation pre-ordering of two channels $\Lambda:\state(\hi^{in})\to \state(\hi^{out})$ and $\Phi:\state(\hi^{in})\to \state(\mathcal{K}^{out})$ is defined as $\Phi \preceq_\mathrm{conc} \Lambda$ if there exists a channel $\Theta\colon\state(\hi^{out})\to\state(\mathcal{K}^{out})$ satisfying $\Phi = \Theta \circ \Lambda$.
    We show that the ordering $\preceq_\mathrm{inc}$ preserves this concatenation pre-ordering $\preceq_\mathrm{conc}$.
    In the following, the Heisenberg picture of a channel $\Lambda$ is represented by $\Lambda^\ast$.  
    
    \begin{prop} \label{prop:conc-inc}
        Let $\bm{\Lambda}=\{\Lambda_1,\dots,\Lambda_n \}$ and $\bm{\Phi}=\{\Phi_1, \dots , \Phi_n\}$ be channel sets, where $\Lambda_i\colon\state(\hi^{in})\to\state(\hi_i^{out})$ and $\Phi_i\colon\state(\hi^{in})\to\state(\mathcal{K}_i^{out})$ for each $i=1,\ldots,n$.
        If $\Phi_i \preceq_\mathrm{conc} \Lambda_i$ for every $i$, then $\bm{\Phi} \preceq_\mathrm{inc} \bm{\Lambda}$ holds.
    \end{prop}
    \begin{proof}
        The assumption means that $\Phi_i =\Theta_i\circ\Lambda_i$ with some channel $\Theta_i\colon\state(\hi_i^{out})\to\state(\mathcal{K}_i^{out})$ for each $i$.
        Consider $\state_0 \subset \state(\hi^{in})$ such that $\bm{\Lambda}$ is $\state_0$-compatible, then we can construct a channel $\Lambda'\colon\state(\hi^{in})\to\state (\hi_1^{out})\otimes\cdots\otimes\state(\hi_n^{out})$ satisfying  
        \[
        \Lambda_i(\rho)=\pTr{1,\ldots,\\i-1,i+1,\ldots,n}{\Lambda'(\rho)}
        \]
        for all $i=1,\ldots,n$ and $\rho \in \state_0$. Let us define a channel $\Phi'\colon \state(\hi^{in})\to\state(\mathcal{K}_1^{out})\otimes\cdots\otimes\state(\mathcal{K}_n^{out})$ by $\Phi' := (\Theta_1 \otimes \cdots \otimes \Theta_n)\circ\Lambda'$ and a state $\tilde{\rho}_i\in\state(\mathcal{K}_i^{out})$ by $\tilde{\rho}_i:=\pTr{1,\ldots,i-1,i+1,\ldots,n}{\Phi'(\rho)}$. We then have 
        \begin{align*}
            \Tr{\tilde{\rho}_i T_i}&=\Tr{\Phi'(\rho) (\1 \otimes \cdots \otimes \1 \otimes T_i\otimes \1\otimes \cdots \otimes \1)}\\
            &= \Tr{\Lambda'(\rho) (\1 \otimes \cdots \otimes \1 \otimes \Theta_i^\ast(T_i)\otimes \1\otimes\cdots\otimes\1)} \\
            &=\Tr{\Lambda_i(\rho)\Theta^*(T_i)}\\
            &= \Tr{\Phi_i(\rho) T_i}
        \end{align*}
        for all $\rho \in \state_0$ and $i=1,\ldots,n,~T_i \in \LL(\hi_i^{out})$. 
        The third equality follows because the trace over $\hi_j^{out}~(j\neq i)$ becomes unity.  
        Since this equation holds for all $T_i\in\LL(\hi_i^{out})$, we have
        \[
        \Phi_i(\rho)=\pTr{1,\ldots,i-1,i+1,\ldots,n}{\Phi'(\rho)} ~(=\tilde{\rho}_i)
        \]
        for all $i=1,\ldots,n$ and $\rho\in\state_0$.
        Hence the channel set $\bm{\Phi}$ is also $\state_0$-compatible. 
    \end{proof}
     
    Finally we consider observable-channel pairs.
    So far we have seen that the incompatibility does not increase under post-processing for observables and concatenations for channels. 
    These properties can be naturally combined for an observable-channel pair.
    To show this, we view an observable $\A=\{\A(x)\}_{x\in X}$ as a quantum-to-classical channel $\Gamma_{\A}\colon\state(\hi^{in})\to\state(\hi_X^{out})$ \cite{Wilde_2017} defined by 
    \[
    \Gamma_{\A}(\rho):=\sum_{x\in X} \Tr{\rho\A(x) } \ketbra{x}{x}.
    \]
    Here $\{\ket{x}|\,x\in X\}$ is an orthonormal basis of $\hi_X^{out} = \ell^2(X)$.
    Our argument relies on the following three known facts \cite{Heinosaari_2017}:
    \begin{itemize}
    \setlength{\parskip}{0cm} 
    \setlength{\itemsep}{0cm} 
        \item[(a)] Let $\{\A,\Lambda\}$ be an observable-channel pair, where $\A=\{\A(x)\}_{x\in X}$ is an observable on $\state(\hi^{in})$ and $\Lambda\colon\state(\hi^{in})\to\state(\hi^{out})$ is a channel. 
        The pair $\{\A,\Lambda\}$ is compatible if and only if the channel pair $\{\Gamma_{\A},\Lambda\}$ is compatible; 
        \item[(b)] For two observables $\A_1$ and $\A_2$ on $\state(\hi^{in})$, the relation $\Gamma_{\A_2} \preceq_\mathrm{conc} \Gamma_{\A_1}$ holds if and only if $\A_2$ is a post-processing of $\A_1$;
        \item[(c)] Let $\bm{\Upsilon}=\{\Upsilon_1,\ldots,\Upsilon_n\}$ and $\bm{\Phi}=\{\Phi_1,\ldots,\Phi_n\}$ be channel sets satisfying $\Phi_i \preceq_\mathrm{conc} \Upsilon_i$ for every $i=1,\dots,n$. If $\bm{\Upsilon}$ is compatible, then $\bm{\Phi}$ is also compatible.
    \end{itemize}
    \noindent Applying these three, we can easily prove the next proposition.
    \begin{prop} \label{prop:obs-channel}
        Let $\A_1,\, \A_2$ be observables and $\Lambda_1,\,\Lambda_2$ be channels. If $\A_2$ is a post-processing of $\A_1$ and $\Lambda_2 \preceq_\mathrm{conc} \Lambda_1$, then $\{\A_2,\Lambda_2\} \preceq_\mathrm{inc} \{\A_1,\Lambda_1\}$ holds. 
    \end{prop}
    \begin{proof}
        According to condition (b), we obtain $ \Gamma_{\A_2} \preceq_\mathrm{conc} \Gamma_{\A_1}$. Hence there exist channels $\Theta_1$ and $\Theta_2$ satisfying $\Gamma_{\A_2} = \Theta_1\circ\Gamma_{\A_1}$ and $\Lambda_2 = \Theta_2\circ\Lambda_1$, respectively. For $\state_0 \subset \state(\hi^{in})$ such that $\{\Gamma_{\A_1},\Lambda_1\}$ is $\state_0$-compatible, we can construct a compatible channel pair $\{ \tilde{\Gamma}_{\A_1},\tilde{\Lambda}_1\}$ that satisfies $\Gamma_{\A_1}(\rho)=\tilde{\Gamma}_{\A_1}(\rho)$ and $\Lambda_1(\rho)=\tilde{\Lambda}_1(\rho)$ for all $\rho\in\state_0$. If we set $\tilde{\Gamma}_{\A_2} := \Theta_1\circ\tilde{\Gamma}_{\A_1}$ and $\tilde{\Lambda}_2 := \Theta_2\circ\tilde{\Lambda}_1$, then $\{ \tilde{\Gamma}_{\A_2}, \tilde{\Lambda}_2\}$ is compatible because of condition (c). The channels $\tilde{\Gamma}_{\A_2}$ and $\tilde{\Lambda}_2$ fulfill $\tilde{\Gamma}_{\A_2}(\rho)= \Gamma_{\A_2}(\rho)$ and $\tilde{\Lambda}_2(\rho)=\Lambda_2(\rho)$ for all $\rho \in \state_0$. Thus $\{\A_2,\Lambda_2\}$ is $\state_0$-compatible. 
    \end{proof}
    
    We have so far investigated fundamental properties of the ordering $\preceq_\mathrm{inc}$. It is natural to introduce an equivalence relation based on this ordering.

    \begin{defi}\label{def_equiv}
        For two device sets $\mathbb{D}=\{\D_1, \ldots, \D_n\}$ and $\mathbb{E}=\{\E_1, \ldots, \E_m\}$, we write $\mathbb{D} \sim_\mathrm{inc} \mathbb{E}$ if both $\mathbb{D} \preceq_\mathrm{inc} \mathbb{E}$ and $\mathbb{E} \preceq_\mathrm{inc} \mathbb{D}$ are satisfied.
    \end{defi}

    The relation $\mathbb{D} \sim_\mathrm{inc}\mathbb{E}$ means that if the incompatibility of $\mathbb{D}$ can be detected using a state set $\state_0$, then the incompatibility of $\mathbb{E}$ also can be detected by the same $\state_0$, and vice versa.
    In this sense, there is no way to operationally distinguish the incompatibility of $\mathbb{D}$ from $\mathbb{E}$.
    Further discussion of this equivalence will be presented in Section \ref{subsec_equivalence}.

\subsection{Relation between operational ordering and distributed sampling}	\label{subsec_dist_sampling}

In this subsection, we demonstrate that the ordering $\preceq_\mathrm{inc}$ that compares incompatibility is related to the performance of observable sets in a specific informational task called distributed sampling \cite{PhysRevA.100.042308}.  
This task consists of the referee and two parties, Alice and Bob. The procedures are as follows: 
\begin{itemize}
    \setlength{\parskip}{0cm} 
    \setlength{\itemsep}{0cm} 
    \item[1.] The referee announces to both parties a state set $\mathcal{S}_0=\{\rho_j\}_{j=1}^m$ and an observable set $\{\A_i\}_{i=1}^n$ with a common outcome set $X$.
    \item[2.] The referee sends a state $\rho_j\in\state_0$ to Alice without informing the index $j$, and a classical label $i\in\{1,\ldots,n\}$ to Bob.
    \item[3.] 
    Alice performs the measurement of some observable or applies some channel on $\rho_j$ and sends its $j$-dependent outcome to Bob as a (classical or quantum) message.
    \item[4.] Bob makes an answer $x\in X$ with probability $P(x|j,i)$ considering the label $i$ from the referee and the $j$-dependent message from Alice.
    \item[5.] The participants repeat these steps. The goal for Alice and Bob is to reproduce the probability distributions $\{ \Tr {\rho_j \A_i(x)}\,\}_{x,j,i}$ using $P(x| j,i)$.
\end{itemize}
Here possible strategies are what observable or channel Alice chooses and how Bob makes an answer $x$.
The collection of probability distributions $\{ \Tr {\rho_j \A_i(x)}\,\}_{x,j,i}$ is called an \textit{$\state_0$-quantum behavior} and is defined as follows.
    \begin{defi}
        Let $X$ be an outcome set and let $i=1,\ldots,n$ and $j=1,\ldots,m$ be classical labels. 
        For a given state set $\state_0=\{\rho_j\}_{j=1}^m$, a collection of conditional probability distributions $\{P(x|j,i)\}_{x,j,i}$ is called an \textit{$\state_0$-quantum behavior} if there exists an observable set $\{\A_i\}_i$ such that 
        \[
        P(x|j,i) = \Tr{\rho_j \A_i(x)},
        \]
        for all $x\in X,\,i=1\ldots,n$ and $j=1,\ldots,m$. 
    \end{defi}

If Alice is allowed to access noiseless quantum communication (QC) to Bob, i.e., Alice can send just her received state $\rho_j$ to Bob without any noise, this task becomes trivial. 
In fact, Bob can reproduce an $\state_0$-quantum behavior $\{ \Tr {\rho_j \A_i(x)}\,\}_{x,j,i}$ by performing the measurement of the announced observable $\A_i$ on a state $\rho_j$.
Instead, we consider the case where only classical communication (CC) is available; i.e., Alice performs some measurement and sends Bob its outcome as a (classical) message. 
An $\state_0$-quantum behavior that can be reproduced under this constraint is called CC-realizable.
If Alice and Bob can reproduce $\state_0$-quantum behaviors that are not CC-realizable, we confirm that there is QC from Alice to Bob. Hence this task works as a QC certifier.

The following proposition establishes the connection between the ordering $\preceq_\mathrm{inc}$ and this task, showing that a ``more incompatible" observable set is a more powerful QC certifier.
The proof is similar to the proof of Theorem 2 in Ref. \cite{PhysRevA.100.042308}. 

\begin{prop} 
    Consider two observable sets $\mathbb{A} = \{\A_i\}_{i=1}^n$ and $\mathbb{B} = \{\B_i\}_{i=1}^n$ with outcome sets X and Y, respectively.  
    For a given state set $\state_0 = \{\rho_j\}_{j=1}^m$, let $\{P_{\mathbb{A}}(x|j,i)\}_{x,j,i}$ and $\{ P_{\mathbb{B}}(y|j,i)\}_{y,j,i}$ be $\state_0$-quantum behaviors such that $P_{\mathbb{A}}(x|j,i) = \Tr{\rho_j \A_i(x)}$ and $P_{\mathbb{B}}(y|j,i) = \Tr{\rho_j \B_i(y)}\,$.
    The following statements are equivalent:
    $\mathrm{(i)}$ $\mathbb{B} \preceq_\mathrm{inc} \mathbb{A}$;
    $\mathrm{(ii)}$ if the $\state_0$-quantum behavior $\{P_{\mathbb{A}}(x|j,i)\}_{x,j,i}$ is CC-realizable, then $\{ P_{\mathbb{B}}(y|j,i)\}_{y,j,i}$ is also CC-realizable for the same $\state_0$. 
\end{prop}
\begin{proof}
    It suffices to show that the $\state_0$-compatibility of $\mathbb{A}$ is equivalent to the CC-realizability of $\{P_{\mathbb{A}}(x|j,i)\}_{x,j,i}$ for the same $\state_0=\{\rho_j\}_{j=1}^m$. Suppose that $\mathbb{A}$ is $\state_0$-compatible. 
    There is a compatible observable set $\tilde{\mathbb{A}}=\{\tilde{\A}_i\}_{i=1}^n$ such that 
    \begin{align}
    \mathrm{Tr}\big[\rho_j\tilde{\A}_i(x)\big]=\mathrm{Tr}\big[\rho_j \A_i(x)\big], \label{cond_tilde_A}
    \end{align}
    for all $x\in  X,\,i=1,\ldots,n$ and $j=1,\ldots,m$.
    Let Alice's strategy be measuring the joint observable of $\tilde{\mathbb{A}}$, then her message to Bob is the outcome $(x_1,\ldots,x_n)\in X\times \cdots \times X$.
    In addition, let Bob's strategy be selecting $x_i$ from $(x_1,\ldots,x_n)$ as his answer.
    Recall that the label $i$ is given in the second step of the task.
    By the definition of $\state_0$-compatibility and \eqref{cond_tilde_A}, this procedure reproduces the behavior $\{P_{\mathbb{A}}(x|j,i)\}_{x,j,i}$, so it is CC-realizable.
    Conversely, suppose that $\{P_{\mathbb{A}}(x|j,i)\}_{x,j,i}$ is CC-realizable for $\state_0$. 
    Then the strategies of Alice and Bob can be written as an observable $\M=\{\M(z)\}$ and a set of conditional functions $\{ h(\cdot |i,z)\}_{i,z}$, respectively, such that
    \[
    \Tr{\rho_j \A_i(x)} = \sum_{z} \Tr{\rho_j \M(z)} h(x|i,z)
    \]
    for all $x\in  X,\,i=1,\ldots,n$ and $j=1,\ldots,m$.
    Define an observable $\tilde{\A}'$ whose outcome is $(x_1,\ldots,x_n)\in X\times\cdots\times X$ by
    \[
    \tilde{\A}'(x_1,\ldots,x_n) = \sum_z \Big( \prod_i h(x_i|i,z)\Big)\M(z).
    \]
    The marginals of $\tilde{\A}'$ form a compatible observable set $\{\tilde{\A}_i\}_i$ satisfying \eqref{cond_tilde_A}, thus $\mathbb{A}$ is $\state_0$-compatible. 
\end{proof}
	
\section{Comparing incompatibility of mutually unbiased qubit observables} \label{sec_qubit}
	This section is concerned with mutually unbiased qubit observables $\mathbb{A}_\mathrm{mub}^{t}=\{\A_\mathrm{ub}^{t\vx},\A_\mathrm{ub}^{t\vy}\}$, which are typical pairs of incompatible observables.
    Recall that an unbiased qubit observable $\A_\mathrm{ub}^\va=\{\A_\mathrm{ub}^\va(\pm)\}$, associated with a vector $\va \in \mathbb{R}^3$ where $|\va| \le 1$, is defined by (see \eqref{def_UQO})
    \[
    \A_\mathrm{ub}^\va(\pm)=\frac{1}{2}(\1\pm\va\cdot\vsigma).
    \]
    We study the incompatibility of mutually unbiased qubit observables through the ordering $\preceq_\mathrm{inc}$.    

\subsection{Operational equivalence of incompatibility} \label{subsec_equivalence}
We first explore the incompatibility of mutually unbiased  qubit observables $\mathbb{A}_\mathrm{mub}^{t}=\{\A_\mathrm{ub}^{t\vx},\A_\mathrm{ub}^{t\vy}\}$ from the perspective of the equivalence relation $\sim_\mathrm{inc}$ (see Definition \ref{def_equiv}).
The main question is how the incompatibility of mutually unbiased qubit observables is characterized in terms of the families of state sets detecting their incompatibility.
To prove the claims in this subsection, it suffices to consider only the state sets $\state_0^{(R)}\subset\state$ defined in \eqref{def_S_R}.
Thus in the following, we concentrate on $\state_0^{(R)}$ rather than all possible sets $\state_0\subset\state$.

Before addressing the main question, we establish the conditions for the incompatibility of two pairs $\mathbb{A}_\mathrm{ub}=\{\A_\mathrm{ub}^{\va_1},\A_\mathrm{ub}^{\va_2}\}$ and $\mathbb{B}_\mathrm{ub}=\{\A_\mathrm{ub}^{\vb_1},\A_\mathrm{ub}^{\vb_2}\}$ of unbiased qubit observables to be equivalent under the relation $\sim_{\mathrm{inc}}$.
Since all compatible pairs are equivalent in terms of $\sim_\mathrm{inc}$, the following discussion focuses on incompatible pairs. 
In addition, for the case of $(\vb_1,\vb_2)=(\pm \va_1, \pm \va_2)$ or $(\vb_1,\vb_2)=(\pm \va_2, \pm \va_1)$, it clearly holds the $\mathbb{A}_\mathrm{ub} \sim_\mathrm{inc} \mathbb{B}_\mathrm{ub}$ because the pairs are essentially identical. 
Then we investigate whether there exist any other pairs that satisfy a (non-trivial) equivalence relation $\mathbb{A}_\mathrm{ub} \sim_\mathrm{inc} \mathbb{B}_\mathrm{ub}$. 

We start with identifying the family of state sets $\state_0^{(R)}$ for which $\mathbb{A}_\mathrm{ub}=\{\A_\mathrm{ub}^{\va_1},\A_\mathrm{ub}^{\va_2}\}$ is $\state_0^{(R)}$-compatible.
Without loss of generality, we assume that the vectors $\va_i\in \R^3~(i=1,2)$ lie in the $xy$ plane and are parametrized by 
\begin{align}
    \va_i =
    \begin{pmatrix}
        a_i \cos\alpha_i\\
        a_i \sin\alpha_i\\
        0
    \end{pmatrix}, \label{parametrize_va_i}
\end{align}
where $\alpha_i\in[0,2\pi)$.
Since $\mathbb{A}_\mathrm{ub}$ is assumed to be incompatible, it follows that (see \eqref{cond_comp_UQO})
 \begin{align} \label{det_incomp_UQO}
	|\va_1 + \va_2| + |\va_1 - \va_2| > 2.
\end{align}
We can represent the state sets $\state_0^{(R)}$ by the associated normal vector $\vn = ( \sin\theta \cos\varphi, \sin\theta \sin\varphi, \cos\theta)^{\mathsf{T}},~\varphi \in [-\pi,\pi),~\theta \in[0,\pi]$ (see \eqref{def_S_R}). Recall that the condition for $\mathbb{A}_\mathrm{ub}$ to be $\state_0^{(R)}$-compatible is given in \eqref{cond_S_R-comp} as
    \begin{align*}
        |P_R(\va_1+\va_2)|+|P_R(\va_1-\va_2)|\leq2.
    \end{align*}  
This condition is transformed into
\begin{equation}
\begin{aligned}
    &f_{\mathbb{A}_\mathrm{ub}}(\varphi,\theta)\\
    &:=L_{\mathbb{A}_\mathrm{ub}}(\varphi)\sin^4 \theta
    + M_{\mathbb{A}_\mathrm{ub}}(\varphi)\sin^2 \theta
    +N_{\mathbb{A}_\mathrm{ub}}\geq 0, \label{def_f_A}
\end{aligned}
\end{equation}
where the coefficients are given by
\begin{equation}\label{eq:coeff}
	\begin{aligned}
		&L_{\mathbb{A}_\mathrm{ub}}(\varphi):=a_1^2 a_2^2 \cos^2(\varphi -\alpha_1) \cos^2(\varphi -\alpha_2),\\
        &M_{\mathbb{A}_\mathrm{ub}}(\varphi):=a_1^2 \cos^2(\varphi -\alpha_1) +a_2^2 \cos^2(\varphi -\alpha_2)\\
		 &\quad-2 a_1^2 a_2^2 \cos \alpha_1 \cos \alpha_2 \cos(\varphi -\alpha_1) \cos(\varphi -\alpha_2),\\
		&N_{\mathbb{A}_\mathrm{ub}}:=a_1^2 a_2^2 \cos^2(\alpha_1 -\alpha_2) -a_1^2 -a_2^2 +1.
	\end{aligned}
\end{equation}
Note that we can always find a state set $\state_0^{(R)}$ for which $\mathbb{A}_\mathrm{ub}$ is $\state_0^{(R)}$-compatible: for example if we set $\state_0^{(R)}$ to the $xz$ plane of the Bloch ball ($(\varphi,\theta)=(\frac{\pi}{2},\frac{\pi}{2})$), then $\mathbb{A}_\mathrm{ub}$ is easily confirmed to satisfy \eqref{cond_S_R-comp}.
We denote by $\mathcal{C}_{\mathbb{A}_\mathrm{ub}}$ the region $(\varphi,\theta)$ where $\mathbb{A}_\mathrm{ub}$ is $\state_0^{(R)}$-compatible, i.e., inequality \eqref{def_f_A} is satisfied. The next proposition is useful for our analysis.

\begin{prop}\label{prop:same_plane}
    Let $\mathbb{A}_\mathrm{ub}=\{\A_\mathrm{ub}^{\va_1},\A_\mathrm{ub}^{\va_2}\}$ and $\mathbb{B}_\mathrm{ub}=\{\A_\mathrm{ub}^{\vb_1},\A_\mathrm{ub}^{\vb_2}\}$ be incompatible pairs of unbiased qubit observables.
    If $\mathbb{A}_\mathrm{ub} \sim_\mathrm{inc} \mathbb{B}_\mathrm{ub}$, then all $\va_1,\va_2,\vb_1,\vb_2$ are in the same 2D plane (the $xy$ plane). 
\end{prop}

\noindent The proof of this proposition relies on the following properties of $f_{\mathbb{A}_\mathrm{ub}}$. 

\begin{lem} \label{prop_of_f_A}
    For an incompatible pair $\mathbb{A}_\mathrm{ub} = \{\A_\mathrm{ub}^{\va_1}, \A_\mathrm{ub}^{\va_2}\}$ of unbiased qubit observables, the function $f_{\mathbb{A}_\mathrm{ub}}$ in \eqref{def_f_A} satisfies the following claims:\\
    $(i)~f_{\mathbb{A}_\mathrm{ub}}(\varphi,\theta) = f_{\mathbb{A}_\mathrm{ub}}(\varphi,\pi-\theta)$.\\
    $(ii)~f_{\mathbb{A}_\mathrm{ub}}(\varphi,\theta) = f_{\mathbb{A}_\mathrm{ub}}(\varphi+\pi,\theta)$.\\
    $(iii)~f_{\mathbb{A}_\mathrm{ub}}(\varphi,\theta_1)<f_{\mathbb{A}_\mathrm{ub}}(\varphi,\theta_2)$ for $0\le\theta_1<\theta_2 \le\frac{\pi}{2}$ and a given $\varphi$.\\
    $(iv)~f_{\mathbb{A}_\mathrm{ub}}(\varphi,0)< 0,\, f_{\mathbb{A}_\mathrm{ub}}(\varphi,\pi)< 0$ and $f_{\mathbb{A}_\mathrm{ub}}(\varphi,\frac{\pi}{2})\geq0$ for all $\varphi$.
\end{lem}
\begin{proof}
    Claims $(i)$ and $(ii)$ are proved by direct calculations. Regarding $(iii)$, we can easily confirm that $L_{\mathbb{A}_\mathrm{ub}}$ is non-negative for all $\varphi$. Furthermore, $M_{\mathbb{A}_\mathrm{ub}}$ is also non-negative for all $\varphi$ because it can be bounded as 
    \begin{align*}
        M_{\mathbb{A}_\mathrm{ub}}(\varphi) &\ge a_1^2 \cos^2(\varphi -\alpha_1) +a_2^2 \cos^2(\varphi -\alpha_2)\\
        &~~ -|2a_1^2 a_2^2 \cos\alpha_1 \cos\alpha_2 \cos(\varphi-\alpha_1) \cos(\varphi-\alpha_2)|\\
        &\ge a_1^2 \cos^2(\varphi -\alpha_1) +a_2^2 \cos^2(\varphi -\alpha_2)\\
        &~~ -|2a_1 a_2 \cos(\varphi-\alpha_1) \cos(\varphi-\alpha_2)|\\
        &=(|a_1 \cos(\varphi-\alpha_1)|-|a_2 \cos(\varphi-\alpha_2|)^2.
    \end{align*} 
     Since $\sin^2 \theta_1<\sin^2 \theta_2$ for $0\le\theta_1<\theta_2\le\frac{\pi}{2}$, combined with $L_{\mathbb{A}_\mathrm{ub}},M_{\mathbb{A}_\mathrm{ub}}\ge 0$, we conclude that $f_{\mathbb{A}_\mathrm{ub}}(\varphi,\theta_1)<f_{\mathbb{A}_\mathrm{ub}}(\varphi,\theta_2)$. 
     For $(iv)$, the first and second relations are verified because the projection $P_R$ becomes the identity for $\theta=0,\,\pi$ and the condition \eqref{det_incomp_UQO} holds.  
     In addition, for $\theta=\frac{\pi}{2}$, the value $|P_R(\va_1 +\va_2)|+|P_R(\va_1 -\va_2)|$ takes either $2|P_R \va_1|$ or $2|P_R \va_2|$, and this implies $f_{\mathbb{A}_\mathrm{ub}}(\varphi,\frac{\pi}{2}) \geq0$. 
\end{proof}

\begin{proof}[Proof of Proposition 6]
    Suppose that $\vb_1$ and $\vb_2$ span a different 2D plane from the $xy$ plane. 
    The relation $\mathbb{A}_\mathrm{ub} \sim_\mathrm{inc} \mathbb{B}_\mathrm{ub}$ implies that the region $\mathcal{C}_{\mathbb{A}_\mathrm{ub}}$ coincides with $\mathcal{C}_{\mathbb{B}_\mathrm{ub}}$. By Lemma \ref{prop_of_f_A} (i), the region $\mathcal{C}_{\mathbb{A}_\mathrm{ub}}$ is symmetric with respect to the $xy$ plane, while the region $\mathcal{C}_{\mathbb{B}_\mathrm{ub}}$ is symmetric with respect to the plane spanned by $\vb_1$ and $\vb_2$. We parametrize the normal vector $\vn_{\mathrm{B}}$ of this plane by $\varphi_\mathrm{B}\in[-\pi,\pi)$ and $\theta_\mathrm{B}\in[0,\pi]$ as
    \[
    \vn_{\mathrm{B}} =
     \begin{pmatrix}
     \sin\theta_{\mathrm{B}} \cos\varphi_{\mathrm{B}}\\
     \sin\theta_{\mathrm{B}} \sin\varphi_{\mathrm{B}}\\
     \cos\theta_{\mathrm{B}} 
     \end{pmatrix}.
    \]
    Consider a state set $\state_0^{(R)}$ satisfying $\varphi=\varphi_{\mathrm{B}}$. It is not difficult to see that the boundary of $\mathcal{C}_{\mathbb{B}_\mathrm{ub}}$, which is symmetric with respect to the plane spanned by $\vb_1$ and $\vb_2$, cannot be symmetric with respect to the $xy$ plane. Therefore, all $\va_1,\,\va_2,\,\vb_1,\,\vb_2$ are in the same plane.
\end{proof}

 Based on this argument, we assume that all $\va_1,\va_2,\vb_1, \vb_2$ are in the $xy$ plane. 
 The following proposition shows that the equivalence relation $\mathbb{A}_\mathrm{ub}\sim_\mathrm{inc}\mathbb{B}_\mathrm{ub}$ does not hold for essentially different $\mathbb{A}_\mathrm{ub}$ and $\mathbb{B}_\mathrm{ub}$ under the constraint $|\va_1|=|\va_2|$ and $|\vb_1|=|\vb_2|$. 

\begin{prop} \label{thm_equivalent}
    Let $\mathbb{A}_\mathrm{ub}=\{\A_\mathrm{ub}^{\va_1},\A_\mathrm{ub}^{\va_2}\}$ and $\mathbb{B}_\mathrm{ub}=\{\A_\mathrm{ub}^{\vb_1},\A_\mathrm{ub}^{\vb_2}\}$ be pairs of unbiased qubit observables satisfying $|\va_1|=|\va_2|=t,~|\vb_1|=|\vb_2|=u$ for some $t,u \in (\frac{1}{\sqrt{2}},1]$. If $\mathbb{A}_\mathrm{ub}\sim_\mathrm{inc}\mathbb{B}_\mathrm{ub}$, then either $(\vb_1,\vb_2)=(\pm\va_1,\pm\va_2)$ or $(\vb_1,\vb_2)=(\pm\va_2,\pm\va_1)$ holds.      
\end{prop}

\noindent 
Before proving this proposition, we first study a more constrained case where $t=u$ and the angle between $\va_1$ and $\va_2$ is equal to the angle between $\vb_1$ and $\vb_2$.

\begin{lem}\label{lem_equiv}
     Let $\mathbb{A}_\mathrm{ub}=\{\A_\mathrm{ub}^{\va_1},\A_\mathrm{ub}^{\va_2}\}$ and $\mathbb{B}_\mathrm{ub}=\{\A_\mathrm{ub}^{\vb_1},\A_\mathrm{ub}^{\vb_2}\}$ 
     satisfy $|\va_1|=|\va_2|=|\vb_1|=|\vb_2|=t~ (t\in(\frac{1}{\sqrt{2}},1])$, and let the angle between $\va_1$ and $\va_2$ be equal to the angle between $\vb_1$ and $\vb_2$. The relation $\mathbb{A}_\mathrm{ub}\sim_\mathrm{inc}\mathbb{B}_\mathrm{ub}$ implies either $(\vb_1,\vb_2)=(\pm\va_1,\pm\va_2)$ or $(\vb_1,\vb_2)=(\pm\va_2,\pm\va_1)$. 
\end{lem}
\begin{proof}
    Suppose that $\mathbb{A}_\mathrm{ub}$ and $\mathbb{B}_\mathrm{ub}$ satisfy $(\vb_1,\vb_2)\neq(\pm\va_1,\pm\va_2)$ and $(\vb_1,\vb_2)\neq(\pm\va_2,\pm\va_1)$.
    The vectors $\va_1,\va_2,\vb_1$ and $\vb_2$ are parametrized similarly to \eqref{parametrize_va_i} by
    \[
    \va_i=
    \begin{pmatrix}
        a_i \cos\alpha_i\\
        a_i \sin\alpha_i\\
        0
    \end{pmatrix}, ~
    \vb_i=
    \begin{pmatrix}
        b_i \cos\beta_i\\
        b_i \sin\beta_i\\
        0
    \end{pmatrix}
    \quad (i=1,2),
    \]
    where $\alpha_i,\beta_i\in[0,2\pi)$.
    Without loss of generality, we can set $\alpha_1,\alpha_2,\beta_1$ and $\beta_2$ as $\alpha_1=0,\alpha_2=\omega,\beta_1=\psi$
    and $\beta_2=\psi+\omega$ with $\omega \in(0,\frac{\pi}{2}]$ and $\psi \in(0,\pi)$. Then the coefficients of $f_{\mathbb{A}_\mathrm{ub}}$ and $f_{\mathbb{B}_\mathrm{ub}}$ defined in \eqref{eq:coeff} become
    \begin{align*}
        &L_{\mathbb{A}_\mathrm{ub}}(\varphi) = t^4 \cos^2\varphi \cos^2(\varphi -\omega), \\
        &M_{\mathbb{A}_\mathrm{ub}}(\varphi) = t^2 \cos^2 \varphi + t^2 \cos^2(\varphi -\omega)\\
        &\hspace{45pt}-2t^4 \cos\omega \cos\varphi \cos(\varphi-\omega), \\
        &L_{\mathbb{B}_\mathrm{ub}}(\varphi) = t^4 \cos^2(\varphi-\psi) \cos^2(\varphi-\psi-\omega),\\
        &M_{\mathbb{B}_\mathrm{ub}}(\varphi) = t^2 \cos^2(\varphi-\psi)+t^2\cos^2(\varphi-\psi-\omega)\\
        &\hspace{30pt} -2t^4\cos\psi \cos(\psi+\omega)\cos(\varphi-\psi)\cos(\varphi-\psi-\omega),\\
        &N_{\mathbb{A}_\mathrm{ub}} = N_{\mathbb{B}_\mathrm{ub}} =t^4\cos^2\omega -2t^2 +1 =: N.
    \end{align*}
    Recall that the relation $\mathbb{A}_\mathrm{ub}\sim_\mathrm{inc}\mathbb{B}_\mathrm{ub}$ implies $\mathcal{C}_{\mathbb{A}_\mathrm{ub}}=\mathcal{C}_{\mathbb{B}_\mathrm{ub}}$. This, combined with the monotonicity proved in Lemma \ref{prop_of_f_A} $(iii)$, ensures that for each $\varphi \in [0,\pi]$ there exists a unique positive simultaneous solution $X_0(\varphi)\in(0,1]$ of the following quadratic equations:
    \begin{equation}\label{eq:equiv_quadratic}
      \begin{aligned}
        L_{\mathbb{A}_\mathrm{ub}}(\varphi)X_0(\varphi)^2 +M_{\mathbb{A}_\mathrm{ub}}(\varphi)X_0(\varphi)+N&=0,\\
         L_{\mathbb{B}_\mathrm{ub}}(\varphi)X_0(\varphi)^2 +M_{\mathbb{B}_\mathrm{ub}}(\varphi)X_0(\varphi)+N&=0.
      \end{aligned}
    \end{equation}
    Consider the specific case $\varphi = \frac{\psi +\omega}{2}$, where $L_{\mathbb{A}_\mathrm{ub}}(\frac{\psi+\omega}{2}) = L_{\mathbb{B}_\mathrm{ub}}(\frac{\psi+\omega}{2})$ holds. 
    For the equations in \eqref{eq:equiv_quadratic} to have a simultaneous solution, the other coefficients must also be equal, i.e., $M_{\mathbb{A}_\mathrm{ub}}(\frac{\psi+\omega}{2})=M_{\mathbb{B}_\mathrm{ub}}(\frac{\psi+\omega}{2})$.
    This can be rewritten as
    \[
    \cos\frac{\psi+\omega}{2} \cos\frac{\psi-\omega}{2} \sin\psi \sin(\psi+\omega)=0.
    \]
    Due to the constraint on $\psi$ and $\omega$, we have
    \begin{align}
        (\beta_2=)~\psi + \omega =\pi. \label{psi_plus_omega}
    \end{align}
    This condition is illustrated in FIG. \ref{fig3}.
    We apply the same argument to the case of $\{\A_\mathrm{ub}^{-\vb_1},\A_\mathrm{ub}^{\va_1}\} \sim_\mathrm{inc} \{\A_\mathrm{ub}^{\vc_1},\A_\mathrm{ub}^{\vc_2}\}$, where $|\vb_1|=|\va_1|=|\vc_1|=|\vc_2|=t$ and the angle between $-\vb_1$ and $\va_1$ is equal to the angle between $\vc_1$ and $\vc_2$.
    Then it follows that $\vc_1=(t \cos(\pi-2\omega),t\sin(\pi-2\omega),0)^\mathsf{T}$ and $\vc_2=\vb_1$ (see FIG. \ref{fig3}). 
    Since the relation $\mathbb{B}_\mathrm{ub}\sim_\mathrm{inc}\{\A_\mathrm{ub}^{-\vb_1},\A_\mathrm{ub}^{\va_1}\}$ holds trivially, we deduce $\mathbb{A}_\mathrm{ub}\sim_\mathrm{inc}\{\A_\mathrm{ub}^{\vc_1},\A_\mathrm{ub}^{\vb_1}\}$.
    However, this contradicts \eqref{psi_plus_omega} unless $\omega = \frac{\pi}{2}$. If $\omega = \frac{\pi}{2}$, the relation $\mathbb{A}_\mathrm{ub}\sim_\mathrm{inc}\{\A_\mathrm{ub}^{\vc_1},\A_\mathrm{ub}^{\vb_1}\}$ becomes $\mathbb{A}_\mathrm{ub}\sim_\mathrm{inc}\mathbb{A}_\mathrm{ub}$, which is trivial. 
\end{proof}    

\begin{figure}
    \centering
    \includegraphics[width=0.7\linewidth]{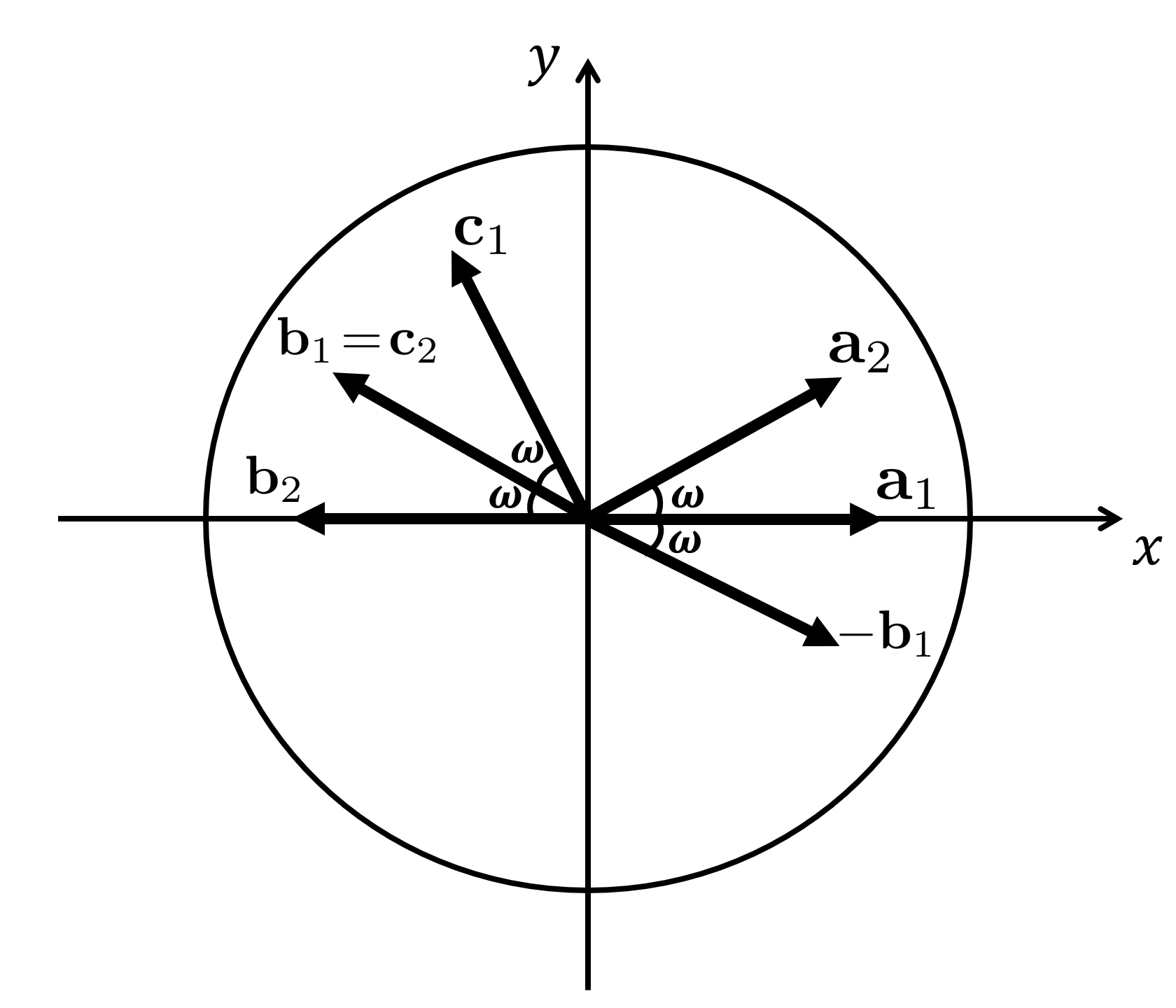}
    \caption{The vectors $\va_1,\va_2,\pm\vb_1,\vb_2,\vc_1$ and $\vc_2$ in the $xy$ plane of the Bloch ball under the condition of \eqref{psi_plus_omega}.}
    \label{fig3}
\end{figure}

\noindent
For the subsequent analysis, it is convenient to introduce an alternative representation of the condition \eqref{def_f_A}:
\begin{align}
    2 &- |\va_1 +\va_2| \sqrt{1 -\sin^2\theta \cos^2(\varphi -\omega_{\va_1 +\va_2})} \notag \\
    &-|\va_1 -\va_2| \sqrt{1-\sin^2\theta \cos^2(\varphi -\omega_{\va_1 -\va_2})} \geq0, \label{def_g_A}
\end{align}
where $\omega_{\va_1 \pm\va_2}$ is the angle between $\va_1\pm\va_2$ and the $x$ axis. This equivalency follows since both conditions are transformed from \eqref{cond_S_R-comp}. 
We denote the left-hand side of \eqref{def_g_A} by $g_{\mathbb{A}_\mathrm{ub}}$ with the substitution $X=\sin^2\theta$, i.e., 
\begin{align*}
 g_{\mathbb{A}_\mathrm{ub}}(\varphi,X) :=\, 2 &- |\va_1 +\va_2| \sqrt{1 -X \cos^2(\varphi -\omega_{\va_1 +\va_2})} \\
    &-|\va_1 -\va_2| \sqrt{1-X \cos^2(\varphi -\omega_{\va_1 -\va_2})}.
\end{align*}
This function inherits the properties of $f_{\mathbb{A}_\mathrm{ub}}$ detailed in Lemma \ref{prop_of_f_A}. Thus we can set $\varphi\in[0,\pi)$ and $X\in[0,1]$

\begin{proof}[Proof of Proposition 7]
    Without loss of generality, we can assume $\alpha_1=0,~\omega_{\va_1+\va_2}\in(0,\frac{\pi}{2})$ and $|\va_1 +\va_2|\geq |\va_1 -\va_2|$ by replacing $\va_2$ with $-\va_2$ if necessary. 
    The core of the proof is again that the boundaries of $\mathcal{C}_{\mathbb{A}_\mathrm{ub}}$ and $\mathcal{C}_{\mathbb{B}_\mathrm{ub}}$ must coincide.
    The assignment of $\mathcal{C}_{\mathbb{A}_\mathrm{ub}}=\mathcal{C}_{\mathbb{B}_\mathrm{ub}}$ led to \eqref{eq:equiv_quadratic} in the proof of Lemma \ref{lem_equiv}. 
    Rephrasing the argument to this case, we obtain the identity
    \begin{align}
         g_{\mathbb{A}_\mathrm{ub}}(\varphi,X_0(\varphi))=g_{\mathbb{B}_\mathrm{ub}}(\varphi,X_0(\varphi))=0. \label{eq:g_A=g_B}
    \end{align}
    For every fixed $X \in(0,1]$, a direct calculation of the partial derivative $\partial g_{\mathbb{A}_\mathrm{ub}}(\varphi,X)/\partial \varphi$ shows that $g_{\mathbb{A}_\mathrm{ub}}(\varphi,X)$ has a maximum at $\varphi=\omega_{\va_1+\va_2}$.
    This fact, combined with the monotonicity of $g_{\mathbb{A}_\mathrm{ub}}(\varphi,X)$ in $X$, implies that the function $X_0(\varphi)$ is minimized at $\varphi=\omega_{\va_1+\va_2}$. 
    The same argument for $\mathbb{B}_\mathrm{ub}$ shows that $X_0(\varphi)$ is also minimized at $\varphi=\omega_{\vb_1+\vb_2}$. 
    Since $\mathcal{C}_{\mathbb{A}_\mathrm{ub}}=\mathcal{C}_{\mathbb{B}_\mathrm{ub}}$, the above discussion leads to $\omega_{\vb_1+\vb_2}=\omega_{\va_1+\va_2}$.
    
    To determine the magnitude of vectors $\va_1,\,\va_2,\,\vb_1$ and $\vb_2$, we evaluate the identity \eqref{eq:g_A=g_B} at two specific angles $\varphi=\omega_{\va_1+\va_2}$ and $\varphi=\omega_{\va_1+\va_2}+\frac{\pi}{4}$. We then obtain the equations
    \begin{align*}
       &t\sin\frac{|\alpha_1-\alpha_2|}{2} -u \sin\frac{|\beta_1-\beta_2|}{2} \\
       &=\! \sqrt{1 \!-\! X_0(\omega_{\va_1+\va_2})} \Big( u \cos\frac{|\beta_1 \!-\! \beta_2|}{2} \!-\! t \cos\frac{|\alpha_1 \!-\! \alpha_2|}{2} \Big) 
    \end{align*}
    and
   \begin{align*}
        t\sin\frac{|\alpha_1-\alpha_2|}{2} &-u \sin\frac{|\beta_1-\beta_2|}{2} \\
        &=u \cos\frac{|\beta_1-\beta_2|}{2} -t \cos\frac{|\alpha_1-\alpha_2|}{2}. 
    \end{align*}
    Since $\mathbb{A}_\mathrm{ub}$ is incompatible, we have $0<X_0(\omega_{\va_1+\va_2})<1$ and thus $\sqrt{1-X_0(\omega_{\va_1+\va_2})} \neq 1$. For the two equations to hold simultaneously, both sides of the equations must be zero. This leads to
    \begin{align}
    \begin{cases}
        u \sin\frac{|\alpha_1-\alpha_2|}{2} -t \sin\frac{|\beta_1-\beta_2|}{2}=0, \\
        u \cos\frac{|\beta_1-\beta_2|}{2} -t \cos\frac{|\alpha_1-\alpha_2|}{2}=0.
    \end{cases} \label{simul_eq}
    \end{align}
    These imply $|\alpha_1-\alpha_2|=|\beta_1-\beta_2|$, and consequently $t=u$. Therefore, the proposition follows from Lemma \ref{lem_equiv}. 
\end{proof}

So far we have established that under the constraint of $|\va_1|=|\va_2|$ and $|\vb_1|=|\vb_2|$, the incompatibility of $\mathbb{A}_\mathrm{ub}=\{\A_\mathrm{ub}^{\va_1},\A_\mathrm{ub}^{\va_2}\}$ and $\mathbb{B}_\mathrm{ub}=\{\A_\mathrm{ub}^{\vb_1},\A_\mathrm{ub}^{\vb_2}\}$ cannot be equivalent except for the trivial cases $(\vb_1,\vb_2)=(\pm\va_1,\pm\va_2)$ or $(\vb_1,\vb_2)=(\pm\va_2,\pm\va_1)$.
Using this result, we now consider the significant case where $\mathbb{A}_\mathrm{ub}$ is mutually unbiased, namely $\mathbb{A}_\mathrm{mub}^t=\{\A_\mathrm{ub}^{t\vx},\A_\mathrm{ub}^{t\vy}\}$ (see Example \ref{eg:incomp_obs}). 
We can prove the same statement without the constraint of $|\vb_1|=|\vb_2|$.

\begin{thm} \label{thm_equiv_txty}
    For mutually unbiased qubit observables $\mathbb{A}_\mathrm{mub}^{t}=\{ \A_\mathrm{ub}^{t\vx},\A_\mathrm{ub}^{t\vy} \}$ and unbiased qubit observables $\mathbb{B}_\mathrm{ub}=\{ \A_\mathrm{ub}^{\vb_1},\A_\mathrm{ub}^{\vb_2} \}$, the relation $\mathbb{A}_\mathrm{mub}^t \sim_\mathrm{inc} \mathbb{B}_\mathrm{ub}$ implies either $(\vb_1,\vb_2)=(\pm t\vx,\pm t\vy)$ or $(\vb_1,\vb_2)=(\pm t\vy,\pm t\vx)$.
\end{thm}

This theorem reveals that the equivalence relation $\sim_\mathrm{inc}$ distinguishes the incompatibility of mutually unbiased qubit observables from all other pairs of unbiased qubit observables. 

To prove this theorem, we need the following lemma.
Here we often write $\overline{\beta}_i=\frac{\pi}{2}-\beta_i$ $(i=1,2)$.
\begin{lem} \label{lem_txty}
	Assume that the pairs of unbiased qubit observables $\mathbb{A}_\mathrm{mub}^t=\{\A_\mathrm{ub}^{t\vx},\A_\mathrm{ub}^{t\vy}\}$ and $\mathbb{B}_\mathrm{ub}=\{\A_\mathrm{ub}^{\vb_1},\A_\mathrm{ub}^{\vb_2}\}$ with $\vb_i=(b_i\cos\beta_i,b_i\sin\beta_i,0)^{\mathsf{T}}~(i=1,2)$ satisfy the relation $\mathbb{A}_\mathrm{mub}^t\sim_\mathrm{inc}\mathbb{B}_\mathrm{ub}$.
	Then $\overline{\mathbb{B}}_\mathrm{ub}=\{\A_\mathrm{ub}^{\overline{\vb}_1},\A_\mathrm{ub}^{\overline{\vb}_2}\}$ with $\overline{\vb}_i=(b_i\cos\overline{\beta}_i,b_i\sin\overline{\beta}_i,0)^{\mathsf{T}}=(b_i\sin\beta_i,b_i\cos\beta_i,0)^{\mathsf{T}}~(i=1,2)$ also satisfies the relation $\mathbb{A}_\mathrm{mub}^t\sim_\mathrm{inc}\overline{\mathbb{B}}_\mathrm{ub}~(\sim_\mathrm{inc}\mathbb{B}_\mathrm{ub})$.
\end{lem}
\begin{proof}
	Define a unitary operator $U$ on $\state$ by 
    \[
    U=\frac{1}{\sqrt{2}}
    \begin{pmatrix}
        0 & 1-i\\
        1+i & 0
    \end{pmatrix}.
    \]
    This operator is self-adjoint and transforms the Pauli operators as $U\sigma_{y}U^\dagger=\sigma_{x}$ and $U\sigma_{x}U^\dagger=\sigma_{y}$.
	By means of this operator $U$, we will prove $\mathbb{A}_\mathrm{mub}^t\sim_\mathrm{inc}\overline{\mathbb{B}}_\mathrm{ub}$.
	Let $\state_0$ be a subset of $\state$ such that $\{\A_\mathrm{ub}^{t\vy},\A_\mathrm{ub}^{t\vx}\}$ is $\state_0$-compatible.
	This implies that we can find a compatible observable set $\{\tilde{\A}_1,\tilde{\A}_2\}$ satisfying 
	\begin{equation}\label{key}
		\begin{aligned}
			&\mathrm{Tr}\Big[\rho\A_\mathrm{ub}^{t\vy}(\pm)\Big]=\mathrm{Tr}\Big[\rho\tilde{\A}_1(\pm)\Big],\\
			&\mathrm{Tr}\Big[\rho\A_\mathrm{ub}^{t\vx}(\pm)\Big]=\mathrm{Tr}\Big[\rho\tilde{\A}_2(\pm)\Big]
		\end{aligned}
	\end{equation}
for all $\rho\in\state_0$. These equations can be rewritten as
	\begin{equation*}\label{key}
	\begin{aligned}
		&\mathrm{Tr}\Big[(U\rho U^\dagger)(U\A_\mathrm{ub}^{t\vy}(\pm)U^\dagger)\Big]=\mathrm{Tr}\Big[(U\rho U^\dagger)(U\tilde{\A}_1(\pm)U^\dagger)\Big],\\
		&\mathrm{Tr}\Big[(U\rho U^\dagger)(U\A_\mathrm{ub}^{t\vx}(\pm)U^\dagger)\Big]=\mathrm{Tr}\Big[(U\rho U^\dagger)(U\tilde{\A}_2(\pm)U^\dagger)\Big],
	\end{aligned}
\end{equation*}
and equivalently,
\begin{equation*}\label{key}
	\begin{aligned}
		&\mathrm{Tr}\Big[(U\rho U^\dagger)\A_\mathrm{ub}^{t\vx}(\pm)\Big]=\mathrm{Tr}\Big[(U\rho U^\dagger)(U\tilde{\A}_1(\pm)U^\dagger)\Big],\\
		&\mathrm{Tr}\Big[(U\rho U^\dagger)\A_\mathrm{ub}^{t\vy}(\pm)\Big]=\mathrm{Tr}\Big[(U\rho U^\dagger)(U\tilde{\A}_2(\pm)U^\dagger)\Big]
	\end{aligned}
\end{equation*}
for all $\rho\in\state_0$.
Since the observable set $\{U\tilde{\A}_1U^\dagger,U\tilde{\A}_2U^\dagger\}$ is compatible, it shows that $\mathbb{A}_\mathrm{mub}^t=\{\A_\mathrm{ub}^{t\vx},\A_\mathrm{ub}^{t\vy}\}$ is $(U\state_0U^\dagger)$-compatible, where $U\state_0 U^\dagger:=\{ U\rho U^\dagger|\, \rho\in\state_0\}$.
Now we can apply the assumption $\mathbb{A}_\mathrm{mub}^t\sim_\mathrm{inc}\mathbb{B}_\mathrm{ub}$ to obtain compatible observables $\{\tilde{\B}_1,\tilde{\B}_2\}$ such that 
\begin{equation*}\label{key}
	\begin{aligned}
		&\mathrm{Tr}\Big[(U\rho U^\dagger)\A_\mathrm{ub}^{\vb_1}(\pm)\Big]=\mathrm{Tr}\Big[(U\rho U^\dagger)\tilde{\B}_1(\pm)\Big],\\
		&\mathrm{Tr}\Big[(U\rho U^\dagger)\A_\mathrm{ub}^{\vb_2}(\pm)\Big]=\mathrm{Tr}\Big[(U\rho U^\dagger)\tilde{\B}_2(\pm)\Big]
	\end{aligned}
\end{equation*}
for all $\rho\in\state_0$. These are equivalent to 
\begin{equation*}\label{key}
	\begin{aligned}
		&\mathrm{Tr}\Big[\rho (U^\dagger\A_\mathrm{ub}^{\vb_1}(\pm)U)\Big]=\mathrm{Tr}\Big[\rho(U^\dagger\tilde{\B}_1(\pm)U)\Big],\\
		&\mathrm{Tr}\Big[\rho(U^\dagger\A_\mathrm{ub}^{\vb_2}(\pm)U)\Big]=\mathrm{Tr}\Big[\rho (U^\dagger\tilde{\B}_2(\pm)U)\Big]
	\end{aligned}
\end{equation*}
for all $\rho\in\state_0$. 
Due to the relation $U^\dagger\A_\mathrm{ub}^{\vb_i}(\pm)U=\A_\mathrm{ub}^{\overline{\vb}_i}(\pm)~(i=1,2)$ and the compatibility of the observable set $\{U^\dagger\tilde{\B}_1U,U^\dagger\tilde{\B}_2U\}$, we conclude that $\overline{\mathbb{B}}_\mathrm{ub}=\{\A_\mathrm{ub}^{\overline{\vb}_1},\A_\mathrm{ub}^{\overline{\vb}_2}\}$
is $\state_0$-compatible and thus $\overline{\mathbb{B}}_\mathrm{ub}\preceq_\mathrm{inc}\{\A_\mathrm{ub}^{t\vy},\A_\mathrm{ub}^{t\vx}\}$.
The other direction $\{\A_\mathrm{ub}^{t\vy},\A_\mathrm{ub}^{t\vx}\}\preceq_\mathrm{inc}\overline{\mathbb{B}}_\mathrm{ub}$ can be proved exactly in the same way.
\end{proof}

\begin{proof}[Proof of Theorem \ref{thm_equiv_txty}]
According to Lemma \ref{lem_txty}, we now have
$\mathbb{A}_\mathrm{mub}^t\sim_\mathrm{inc}\mathbb{B}_\mathrm{ub}\sim_\mathrm{inc}\overline{\mathbb{B}}_\mathrm{ub}$.
It implies that, for any $\varphi\in[0,\pi]$, there is a unique simultaneous solution $X_0(\varphi)\in(0,1]$ for the following quadratic equations 
\begin{align}
		&L_{\mathbb{A}_\mathrm{mub}^t}(\varphi)X_0 (\varphi)^2
		\!+ M_{\mathbb{A}_\mathrm{mub}^t}(\varphi)X_0(\varphi)
		\!+N_{\mathbb{A}_\mathrm{mub}^t}\!=0,\label{eq:X(phi)_A} \\
		&L_{\mathbb{B}_\mathrm{ub}}(\varphi)X_0 (\varphi)^2
		+ M_{\mathbb{B}_\mathrm{ub}}(\varphi)X_0(\varphi)
		+N_{\mathbb{B}_\mathrm{ub}}=0,\label{eq:X(phi)_B} \\
		&L_{\overline{\mathbb{B}}_\mathrm{ub}}(\varphi)X_0 (\varphi)^2
		+ M_{\overline{\mathbb{B}}_\mathrm{ub}}(\varphi)X_0(\varphi)
		+N_{\overline{\mathbb{B}}_\mathrm{ub}}=0
\end{align}
as in the proof of Lemma \ref{lem_equiv}.
Here the coefficients are defined in \eqref{eq:coeff}.
In the case of $\mathbb{A}_\mathrm{mub}^t=\{\A^{t\vx},\A^{t\vy}\}$, the equation \eqref{eq:X(phi)_A} for $X_0(\varphi)$ becomes
\begin{equation*}
	\big[t^4 \cos^2\varphi \sin^2\varphi\big]X_0(\varphi)^2
	+t^2X_0(\varphi)+1-2t^2=0.
	\label{eq:MUt}
\end{equation*}
Since the coefficient $\cos^2\varphi \sin^2\varphi$ of the initial term satisfies $\cos^2\varphi \sin^2\varphi=\cos^2(\varphi+\frac{\pi}{2})\sin^2(\varphi+\frac{\pi}{2})=\cos^2(-\varphi)\sin^2(-\varphi)$, it holds that $X_0(\varphi)=X_0(\varphi+\frac{\pi}{2})=X_0(-\varphi)$ for all $\varphi\in[0,\pi]$.
In particular, we have $X_0(\varphi)=X_0(\frac{\pi}{2}-\varphi)$, and this, together with the properties $L_{\overline{\mathbb{B}}_\mathrm{ub}}(\varphi)=L_{\mathbb{B}_\mathrm{ub}}(\frac{\pi}{2}-\varphi)$ and $N_{\overline{\mathbb{B}}_\mathrm{ub}}(\varphi)=N_{\mathbb{B}_\mathrm{ub}}(\frac{\pi}{2}-\varphi)$, leads to
\begin{align*}
	M_{\overline{\mathbb{B}}_\mathrm{ub}}(\varphi)= M_{\mathbb{B}_\mathrm{ub}}\Big(\frac{\pi}{2}-\varphi\Big)
\end{align*}
for all $\varphi\in[0,\pi]$.
It can be simplified as
\begin{align*}
	\sin(\varphi+\beta_1)\sin(\varphi+\beta_2)\cos(\beta_1+\beta_2)=0
\end{align*}
for all $\varphi\in[0,\pi]$. Thus 
\begin{equation}\label{key}
	\cos(\beta_1+\beta_2)=0,
\end{equation}
which yields
\begin{align}
	\beta_1+\beta_2=\frac{\pi}{2},~\frac{3\pi}{2}.
\end{align}
If $\beta_1 + \beta_2 = \frac{\pi}{2}$, then $X_0(\varphi)=X_0(\frac{\pi}{2}-\varphi)$ implies $X_0(\beta_1)=X_0(\beta_2)$. Assigning $\varphi=\beta_1,\beta_2$ to \eqref{eq:X(phi)_B}, we obtain
\begin{align*}
    L_{\mathbb{B}_\mathrm{ub}}(\beta_1)X_0 (\beta_1)^2+M_{\mathbb{B}_\mathrm{ub}}(\beta_1)X_0(\beta_1)+N_{\mathbb{B}_\mathrm{ub}}=0,\\
     L_{\mathbb{B}_\mathrm{ub}}(\beta_2)X_0 (\beta_1)^2+M_{\mathbb{B}_\mathrm{ub}}(\beta_2)X_0(\beta_1)+N_{\mathbb{B}_\mathrm{ub}}=0,
\end{align*}
respectively. Since $L_{\mathbb{B}_\mathrm{ub}}(\beta_1)=L_{\mathbb{B}_\mathrm{ub}}(\beta_2)$, we have $M_{\mathbb{B}_\mathrm{ub}}(\beta_1)=M_{\mathbb{B}_\mathrm{ub}}(\beta_2)$, which means $(1-\cos^2(\beta_1-\beta_2))(b_1^2 -b_2^2)=0$. The case of $\cos^2(\beta_1-\beta_2)= 1$ contradicts \eqref{det_incomp_UQO}, so $b_1 = b_2$ holds. 
If $\beta_1+\beta_2=\frac{3\pi}{2}$, we can apply the same discussion because $X_0(\varphi)=X_0(\varphi+\frac{\pi}{2})=X_0(-\varphi)$ implies $X_0(\beta_1) =X_0(\beta_2)$. 
According to Proposition \ref{thm_equivalent}, the condition $b_1=b_2$ leads to $(\vb_1,\vb_2)=(\pm t\vx,\pm t\vy)$ or $(\vb_1,\vb_2)=(\pm t\vy,\pm t\vx)$.
\end{proof}

\subsection{Numerical comparison of incompatibility} \label{subsec_numerical_analysis}
    We have shown that the incompatibility of mutually unbiased qubit observables is uniquely characterized within the pairs of unbiased qubit observables in the sense of the relation $\sim_\mathrm{inc}$. 
    The focus now shifts to studying when the ordering $\preceq_\mathrm{inc}$ holds among mutually unbiased qubit observables.
    We prepare mutually unbiased qubit observables $\mathbb{A}_\mathrm{mub}^{t=1}$ and $\mathbb{B}_\mathrm{mub}^t(\theta)$ as
    \begin{align*}
        \mathbb{A}_\mathrm{mub}^{t=1} &= \{\A^\vx,\A^\vy\},\\
        \mathbb{B}_\mathrm{mub}^t(\theta) &= \{\A^{t\vx},\A^{t(\vy \cos\theta + \vz \sin\theta)}\}, 
    \end{align*}
    where $\frac{1}{\sqrt{2}}< t\le 1$ and $0\le \theta \le \frac{\pi}{2}$.
    We are interested in whether the ordering $\preceq_\mathrm{inc}$ establishes novel classifications of device sets beyond Propositions \ref{fund_features} through \ref{prop:obs-channel}.
    In this case, our goal is to determine if there exist parameter regions $(t,\theta)$ satisfying $\mathbb{B}_\mathrm{mub}^t(\theta) \preceq_\mathrm{inc} \mathbb{A}_\mathrm{mub}^{t=1}$ beyond the consequence of Proposition \ref{fund_features} and \ref{prop:post-processing}.

    We begin with identifying the region $(t,\theta)$ where $\mathbb{B}_\mathrm{mub}^t(\theta) \preceq_\mathrm{inc} \mathbb{A}_\mathrm{mub}^{t=1}$ is confirmed by the conventional relations given in Proposition \ref{fund_features} and \ref{prop:post-processing}. 
    If $\mathbb{B}_\mathrm{mub}^t(\theta)$ is a convex combination of $\mathbb{A}_\mathrm{mub}^{t=1}$ and a compatible observable set (see Proposition \ref{fund_features}), we can find a real number $0 \leq \lambda \leq 1$ and a compatible pair $\mathbb{N}=\{\mathsf{N}_1, \mathsf{N}_2\}$ of binary observables satisfying
    \begin{equation}
    \begin{aligned}\label{eq:conv_comb_UQO}
    \A^{t\vx}(+) &= \lambda \A^{\vx}(+) + (1-\lambda)\mathsf{N}_1(+),\\
    \A^{t(\vy \cos\theta + \vz \sin\theta)}(+) &= \lambda \A^{\vy}(+) + (1-\lambda)\mathsf{N}_2(+). 
    \end{aligned}
    \end{equation}
    The extreme case of $\lambda=1$ corresponds to the point $(t,\theta)=(1,0)$.
    For the other cases, let us parametrize $\mathsf{N}_1(+)$ and $\mathsf{N}_2(+)$ similarly to \eqref{def_QO} as
    \[
    \mathsf{N}_i(+)=\frac{1}{2}(n_i^0\1+\vn_i \cdot \vsigma)\quad(i=1,2),
    \]
    where $\vn_i\in\R^3$ is a vector such that $|\vn_i|\le n_i^0\le2-|\vn_i|$.
    Then the relations \eqref{eq:conv_comb_UQO} can be expressed as
    \begin{align*}
        n_1^0&=1,\, \vn_1=\Big(\frac{t-\lambda}{1-\lambda},0,0\Big)^\mathsf{T},\\
        n_2^0&=1,\, \vn_2=\Big(0,\frac{t\cos\theta-\lambda}{1-\lambda},\frac{t\sin\theta}{1-\lambda}\Big)^\mathsf{T}.
    \end{align*}
    Since $\mathbb{N}=\{\mathsf{N}_1,\mathsf{N}_2\}$ is compatible, the vectors $\vn_1$ and $\vn_2$ satisfy the condition \eqref{cond_comp_UQO}, so we have 
    \[
        \frac{1}{t} -1 + \sqrt{2 -\frac{1}{t^2}} \leq \cos\theta ~ \left(\frac{1}{\sqrt{2}} < t \leq 1,\, 0\leq \theta \leq \frac{\pi}{2}\right).
    \]
    For the parameters $(t,\theta)$ satisfying this inequality, we can verify $\mathbb{B}_\mathrm{mub}^t(\theta) \preceq_\mathrm{inc} \mathbb{A}_\mathrm{mub}^{t=1}$ from Proposition \ref{fund_features}. These cases are illustrated by the blue region in FIG. \ref{fig4}.
    Additionally, if $\mathbb{A}_\mathrm{mub}^{t=1}$ and $\mathbb{B}_\mathrm{mub}^t(\theta)$ fulfill the assumption of Proposition \ref{prop:post-processing}, we obtain $\theta=0$ by direct calculations. This case is already analyzed in Proposition \ref{fund_features}.

    Our interest is the existence of classifications that are not explained by Proposition \ref{fund_features} and \ref{prop:post-processing}.
    Specifically, the question is if there exists $(t,\theta)$ for which $\mathbb{B}_\mathrm{mub}^t(\theta) \preceq_\mathrm{inc} \mathbb{A}_\mathrm{mub}^{t=1}$ holds outside the blue region in FIG. \ref{fig4}. 
    To answer this question, we performed a numerical analysis. 
    By using the method described in Example \ref{eg:S0-incomp_qubit_obs}, we can judge the $\state_0$-incompatibility of $\mathbb{A}_\mathrm{mub}^{t=1}$ and $\mathbb{B}_\mathrm{mub}^t(\theta)$ for a fixed $\state_0 \subset \state$.
    Furthermore, for each $(t,\theta)$ we confirm $\mathbb{B}_\mathrm{mub}^t(\theta) \npreceq_\mathrm{inc} \mathbb{A}_\mathrm{mub}^{t=1}$ if we find a state set $\state_0$ such that $\mathbb{B}_\mathrm{mub}^t(\theta)$ is $\state_0$-incompatible but $\mathbb{A}_\mathrm{mub}^{t=1}$ is $\state_0$-compatible.
    We searched numerically for such $\state_0$ from all state sets. Recall that it suffices to investigate two types $\state_0^{(3)}$ and $\state_0^{(2)}$ of state sets (see Example \ref{eg:S0-incomp_qubit_obs}). 
    
    Before considering general $\state_0^{(3)}$, we focus on a special type $\state_0^{(R)}$ defined in \eqref{def_S_R}.
    Thanks to \eqref{cond_S_R-comp}, the conditions for $\mathbb{A}_\mathrm{mub}^{t=1}$ and $\mathbb{B}_\mathrm{mub}^t(\theta) $ to be $\state_0^{(R)}$-incompatible are given as
    \begin{align}
        1-n_y^2-n_z^2&< \frac{2}{n_y^2+1}-1, \label{cond:S_R-comp_A}\\
        1-n_y^2-n_z^2&< \frac{2}{t^2(n_y \cos\theta+n_z \sin\theta)^2+1}-\frac{1}{t^2}, \label{cond:S_R-comp_B}
    \end{align}
    respectively.
    Here $\vn=(n_x,n_y,n_z)^\mathsf{T}$ is the normal vector of $R$ (see \eqref{def_S_R}).
    We searched for a state set $\state_0^{(R)}$ that violates \eqref{cond:S_R-comp_A} but satisfies \eqref{cond:S_R-comp_B} from all grid points of $n_y$ and $n_z$.
    We investigated this for all grid points of $(t,\theta)$.
    In FIG. \ref{fig4}, the gray region illustrates $(t,\theta)$ where $\mathbb{B}_\mathrm{mub}^t(\theta) \npreceq_\mathrm{inc} \mathbb{A}_\mathrm{mub}^{t=1}$ is confirmed by $\state_0^{(R)}$. 

    \begin{figure}
        \centering
        \includegraphics[width=0.95\linewidth]{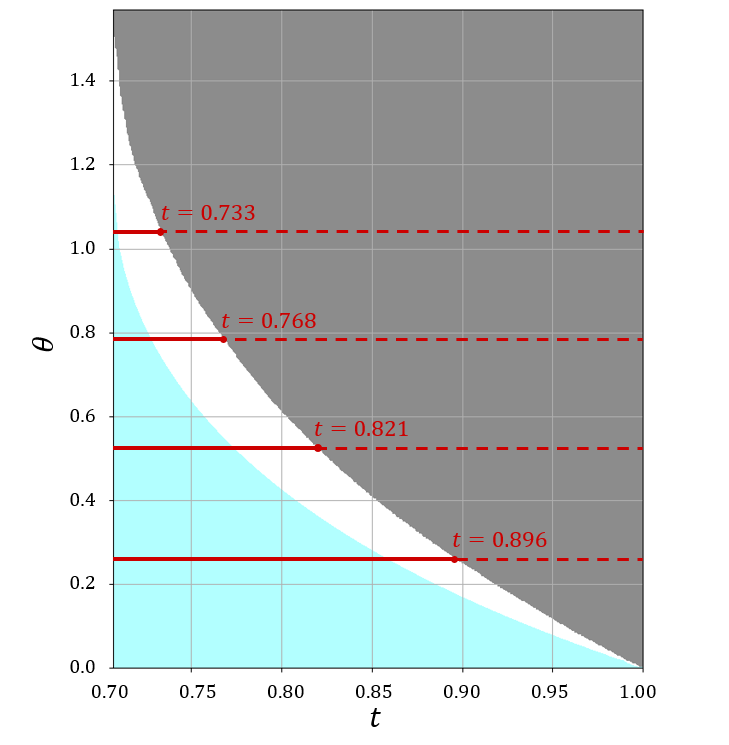}
        \caption{The region $(t,\theta)$ where the ordering $\mathbb{B}_\mathrm{mub}^t(\theta) \preceq_\mathrm{inc} \mathbb{A}_\mathrm{mub}^{t=1}$ holds for mutually unbiased qubit observables $\mathbb{A}_\mathrm{mub}^{t=1} = \{\A^\vx,\A^\vy\}$ and $\mathbb{B}_\mathrm{mub}^t(\theta) = \{\A^{t\vx},\A^{t(\vy \cos\theta + \vz \sin\theta) }\}$. 
            \textbf{Blue region}: The ordering $\mathbb{B}_\mathrm{mub}^t(\theta) \preceq_\mathrm{inc} \mathbb{A}_\mathrm{mub}^{t=1}$ is explained by Proposition \ref{fund_features} and \ref{prop:post-processing}.
            \textbf{Gray region}: We confirm $\mathbb{B}_\mathrm{mub}^t(\theta) \npreceq_\mathrm{inc} \mathbb{A}_\mathrm{mub}^{t=1}$ through only analyzing $\state_0^{(R)}$. 
            \textbf{White region}: We need to investigate general $\state_0$ to judge whether the ordering holds.
            \textbf{Solid lines}: For $\theta=\pi/12,\pi/6,\pi/4,\pi/3$, we newly found that $\mathbb{B}_\mathrm{mub}^t(\theta) \preceq_\mathrm{inc} \mathbb{A}_\mathrm{mub}^{t=1}$ holds. The maximum value of $t$ that realizes $\mathbb{B}_\mathrm{mub}^t(\theta) \preceq_\mathrm{inc} \mathbb{A}_\mathrm{mub}^{t=1}$ for each $\theta$ lies in the boundary of the white and the gray regions.}
        \label{fig4}
    \end{figure}

    The remaining region (the white region in FIG. \ref{fig4}) represents the cases where we have to analyze state sets other than $\state_0^{(R)}$ to determine whether $\mathbb{B}_\mathrm{mub}^t(\theta) \preceq_\mathrm{inc} \mathbb{A}_\mathrm{mub}^{t=1}$ holds. To address this probrem, we now turn to general state sets $\state_0^{(3)}$ and $\state_0^{(2)}$.
    For a fixed $\state_0^{(3)}$ parametrized by $r$ and $\vn$ (see FIG. \ref{fig1}), the conditions for $\mathbb{A}_\mathrm{mub}^{t=1}$ and $\mathbb{B}_\mathrm{mub}^t(\theta)$ to be $\state_0^{(3)}$-incompatible are given by \eqref{cond_S0-comp_(3)} as
    \[
        \min _{\lambda_1,\lambda_2}F_{\mathbb{A}_\mathrm{mub}^{t=1}}^{(r,\vn)}(\lambda_1,\lambda_2) > 0
    \]
    and
    \[
        \min _{\lambda_1,\lambda_2}F_{\mathbb{B}_\mathrm{mub}^t(\theta)}^{(r,\vn)}(\lambda_1,\lambda_2) > 0,
    \]
    respectively. Thus the condition for $\mathbb{B}_\mathrm{mub}^t(\theta) \npreceq_\mathrm{inc} \mathbb{A}_\mathrm{mub}^{t=1}$ is
    \begin{align}
        \max_{r,\vn} \min _{\lambda_1,\lambda_2}F_{\mathbb{B}_\mathrm{mub}^t(\theta)}^{(r,\vn)}(\lambda_1,\lambda_2) > 0, \label{maxminF_(3)}
    \end{align}
    where the maximization is over all $(r,\vn)$ characterizing $\state_0^{(3)}$ for which $\mathbb{A}_\mathrm{mub}^{t=1}$ is $\state_0^{(3)}$-compatible.
    Similarly, for a fixed $\state_0^{(2)}$ parametrized by $r,\, \vn$ and $\vm$ (see FIG. \ref{fig2}), the condition for $\mathbb{B}_\mathrm{mub}^t(\theta) \npreceq_\mathrm{inc} \mathbb{A}_\mathrm{mub}^{t=1}$ is 
    \begin{align}
        \max_{r,\vn,\vm}  \min _{\substack{\lambda_1,\lambda_2,\\ \xi_1,\xi_2}}F_{\mathbb{B}_\mathrm{mub}^t(\theta)}^{(r,\vn,\vm)}(\lambda_1,\lambda_2,\xi_1,\xi_2) > 0, \label{maxminF_(2)}
    \end{align}
    where the maximization is over all $(r,\vn,\vm)$ characterizing $\state_0^{(2)}$ for which $\mathbb{A}_\mathrm{mub}^{t=1}$ is $\state_0^{(2)}$-compatible.
    For each $(t,\theta)$ in the white region of FIG. \ref{fig4}, if both \eqref{maxminF_(3)} and \eqref{maxminF_(2)} are violated, then we can verify that $\mathbb{B}_\mathrm{mub}^t(\theta) \preceq_\mathrm{inc} \mathbb{A}_\mathrm{mub}^{t=1}$ holds. 
    We focused on four representative values $\theta=\pi/12,\pi/6,\pi/4,\pi/3$ and numerically analyzed these conditions.
    The minimization of \eqref{maxminF_(3)} and \eqref{maxminF_(2)} was computed by solving Sequential Least Squares Programming (SLSQP) \cite{kraft1988software}, and the maximization was performed over all grid points of $r,\vn$ (and $\vm$).
    The maximum values of $t$ for which $\mathbb{B}_\mathrm{mub}^t(\theta) \preceq_\mathrm{inc} \mathbb{A}_\mathrm{mub}^{t=1}$ holds for each $\theta$ are listed in TABLE \ref{table1}. 
    Here we can conclude that the ordering relation $\mathbb{B}_\mathrm{mub}^t(\theta) \preceq_\mathrm{inc} \mathbb{A}_\mathrm{mub}^{t=1}$ is established beyond the previously derived cases (the blue region in FIG. \ref{fig4}).
    The regions $(t,\theta)$ where we newly revealed $\mathbb{B}_\mathrm{mub}^t(\theta) \preceq_\mathrm{inc} \mathbb{A}_\mathrm{mub}^{t=1}$ are illustrated in FIG. \ref{fig4} as an intersection of the solid line and the white region.
    Notably, the maximum values of $t$ satisfying $\mathbb{B}_\mathrm{mub}^t(\theta) \preceq_\mathrm{inc} \mathbb{A}_\mathrm{mub}^{t=1}$ coincide with the boundary between the white and gray regions.
    This shows that we can judge whether the ordering establishes through only state sets $\state_0^{(R)}$ for the cases of $\theta=\pi/12,\pi/6,\pi/4,\pi/3$.
    Moreover, we conjecture that $\mathbb{B}_\mathrm{mub}^t(\theta) \preceq_\mathrm{inc} \mathbb{A}_\mathrm{mub}^{t=1}$ holds for all $(t,\theta)$ in the white region of FIG. \ref{fig4}.
    If this conjecture is true, it is enough to examine the state sets $\state_0^{(R)}$, which are easier to investigate than the other types of $\state_0$, to determine whether the ordering holds. 
    Recall that our previous result Theorem \ref{thm_equiv_txty} was derived also from the analysis of $\state_0^{(R)}$.
    Thus, we expect $\state_0^{(R)}$ to be a crucial type of state sets to study the incompatibility of mutually unbiased qubit observables.

\begin{table}[h]
\caption{The maximum values of $t$ for which $\mathbb{B}_\mathrm{mub}^t(\theta) \preceq_\mathrm{inc} \mathbb{A}_\mathrm{mub}^{t=1}$ holds for each $\theta$. These values coincide with the boundary of the gray region in FIG. \ref{fig4}. }
\vspace{0.2cm}
\begin{tabular}{
>{\columncolor[HTML]{FFFFFF}}c |
>{\columncolor[HTML]{FFFFFF}}c 
>{\columncolor[HTML]{FFFFFF}}c 
>{\columncolor[HTML]{FFFFFF}}c 
>{\columncolor[HTML]{FFFFFF}}c }
 $\quad\theta\quad$ & $\quad\pi/12\quad$ & $\quad\pi/6\quad$ & $\quad\pi/4\quad$ & $\quad\pi/3\quad$ \\ \hline
$t$ & 0.896 & 0.821 & 0.768 & 0.733
\end{tabular}
\label{table1}
\end{table}

	\section{Conclusion} \label{sec_conclusion}
In this paper, we have introduced an operational ordering $\preceq_\mathrm{inc}$ to compare the incompatibility of device sets.
This ordering is based on the ease of detecting incompatibility when the available states are restricted.
We have shown that the ordering is not only consistent with the fundamental properties of incompatibility but also establishes novel, more fine-grained hierarchies.
These hierarchies were demonstrated through a typical class of incompatibility exhibited by mutually unbiased qubit observables.
Furthermore, we studied the equivalence relation $\sim_\mathrm{inc}$ induced by the ordering.
Our analysis found that the incompatibility of mutually unbiased qubit observables is distinct from all other pairs of unbiased qubit observables in terms of the relation $\sim_\mathrm{inc}$.
We also revealed that the ordering plays a significant role in distributed sampling, which serves as a certifier for quantum communication.

Concrete analysis of higher-dimensional systems, as well as a larger number or other types of devices, will provide further insight.
However, we have to overcome several difficulties to investigate these situations.
For instance, the counterparts of $\tilde{\A}_i$ in \eqref{tilde_A} become quite complicated, and the necessary and sufficient conditions of incompatibility for such situations are unknown.
In these cases, semidefinite programming (SDP) may offer a possible avenue for future investigations. 
Additionally, it is worth noting that the operational ordering here can be introduced naturally in the framework of general probabilistic theories (GPTs) \cite{Plavala2023,Lami2017,Takakura2022,Wilce2025GPT}.
Devices are basic elements also in the realm of GPTs, and their incompatibility has been actively investigated \cite{PhysRevLett.99.240501,PhysRevA.94.042108,PhysRevA.95.032127,PhysRevA.98.012133,kuramochi2020compatibility,Takakura2020,Takakura_2021,Plavala_2022,PhysRevLett.128.160402,PhysRevA.110.062210,m7ln-tb1s}.
It will be an interesting problem to study whether our way of comparing incompatibility reveals structures specific to quantum incompatibility.
Future work will also address the generalization of Proposition \ref{thm_equivalent}, as it remains an open problem whether the result holds without the constraints of $|\va_1|=|\va_2|$ and $|\vb_1|=|\vb_2|$.
Another promising direction is the application of the ordering to specific no-go theorems such as the no-broadcasting theorem.
This line of research may reveal insights into both quantum theory and GPTs.\\

\section*{Acknowledgment}
We would like to thank Takayuki Miyadera and Gen Kimura for fruitful discussions.
We also appreciate the anonymous referee.
K.T. acknowledges support from JST SPRING Grant No. JPMJSP2125.
R.T. acknowledges support from JST COI-NEXT program Grant No. JPMJPF2014 and JSPS KAKENHI Grant No. JP25K17314.


%

\end{document}